\documentclass[11pt,a4paper,reqno]{amsart}
\usepackage{amscd,amsthm,amsmath,amsfonts,amssymb,amsxtra,appendix,bookmark,dsfont,latexsym,bm,hyperref,delarray,color,euscript,amsgen,amsbsy,amsopn,amscd,latexsym,mathrsfs,bbm}

\makeatletter

\usepackage{hyperref}

\makeatletter
\def\@tocline#1#2#3#4#5#6#7{\relax
  \ifnum #1>\c@tocdepth 
  \else
    \par \addpenalty\@secpenalty\addvspace{#2}%
    \begingroup \hyphenpenalty\@M
    \@ifempty{#4}{%
      \@tempdima\csname r@tocindent\number#1\endcsname\relax
    }{%
      \@tempdima#4\relax
    }%
    \parindent\z@ \leftskip#3\relax \advance\leftskip\@tempdima\relax
    \rightskip\@pnumwidth plus4em \parfillskip-\@pnumwidth
    #5\leavevmode\hskip-\@tempdima
      \ifcase #1
       \or\or \hskip 1em \or \hskip 2em \else \hskip 3em \fi%
      #6\nobreak\relax
      \dotfill
      \hbox to\@pnumwidth{\@tocpagenum{#7}}
    \par
    \nobreak
    \endgroup
  \fi}
\makeatother

\newtheorem{theorem}{Theorem}[section]
\newtheorem{lemma}[theorem]{Lemma}
\newtheorem{proposition}[theorem]{Proposition}
\newtheorem{assumption}[theorem]{Asumption}

\theoremstyle{definition}

\theoremstyle{remark}
\newtheorem{remark}[theorem]{Remark}

\newcommand\R{{\ensuremath {\mathbb R} }}

\renewcommand\phi{\varphi}

\newcommand{\wto}{\rightharpoonup}

\newcommand{\cP}{\mathcal{P}}

\newcommand{\cF}{\mathcal{F}}

\newcommand{\cZ}{\mathcal{Z}}

\renewcommand{\epsilon}{\varepsilon}

\renewcommand{\ge}{\geqslant}
\renewcommand{\le}{\leqslant}
\renewcommand{\geq}{\geqslant}
\renewcommand{\leq}{\leqslant}

\newcommand{\eps}{\varepsilon}

\newcommand{\im}{\mathrm{i}}
\newcommand{\Eflo}{E^{\mathrm{flo}}}
\newcommand{\yv}{\widehat{y}}

\newcommand{\Emp}{\mathrm{Emp}}

\numberwithin{equation}{section}

\begin{document}

\title{Stability of the Laughlin phase against long-range interactions}

\author[A. Olgiati]{Alessandro Olgiati}
\address{Universit\'e Grenoble-Alpes \& CNRS,  LPMMC, F-38000 Grenoble, France}
\email{alessandro.olgiati@lpmmc.cnrs.fr}

\author[N. Rougerie]{Nicolas Rougerie}
\address{Universit\'e Grenoble-Alpes \& CNRS,  LPMMC, F-38000 Grenoble, France}
\email{nicolas.rougerie@lpmmc.cnrs.fr}

\date{February, 2020}

\begin{abstract}
A natural, ``perturbative'', problem in the modelization of the fractional quantum Hall effect is to minimize a classical energy functional within a variational set based on Laughlin's wave-function. We prove that, for small enough pair interactions, and asymptotically for large particle numbers, a minimizer can always be looked for in the particular form of uncorrelated quasi-holes superimposed to Laughlin's wave-function. 
\end{abstract}

\maketitle

\tableofcontents

\section{Introduction}

The fractional quantum Hall effect (FQHE) occurs when a gas of interacting electrons is confined to two space dimensions and subjected to a strong external magnetic field, perpendicular to the confinement plane~\cite{Jain-07,Girvin-04,Goerbig-09,StoTsuGos-99,Laughlin-99}. Under these conditions, the Hall resistance exhibits a quantized behavior by forming a series of plateaus in correspondence to particular rational values of the filling factor\footnote{$\rho$ is the electron density, $B$ the magnetic field, the rest universal constants.} 
$$\nu = \frac{hc}{e_0} \frac{\rho}{B}.$$
While the integer quantum Hall effect (IQHE, integer values of the filling factor) can, in first approximation, be understood in terms of non-interacting electrons, the origin of the FQHE (special fractional values of the filling factor) lies in inter-particle interactions. It was first recognized by Laughlin~\cite{Laughlin-83,Laughlin-87} that the origin of the FQHE is a very remarkable strongly correlated phase that the gas exhibits in its ground state, for certain special values of $\nu$. It was then proposed, and later confirmed experimentally~\cite{SamGlaJinEti-97,MahaluEtal-97,YacobiEtal-04} that within this phase the gas presents quasiparticle excitations with fractional charge. They are also believed to have fractional statistics~\cite{AroSchWil-84,LunRou-16}.

At the basis of our theoretical understanding of the FQHE is Laughlin's wave-function which, for a system with filling factor $\nu = 1/\ell$ with $\ell\in\mathbb{N}$, and subjected to an external magnetic field $B$, reads
\begin{equation} \label{eq:laughlin}
\Psi_{\mathrm{Lau}}(z_1,\dots,z_N)=c_{\mathrm{Lau}}\prod_{i<j}(z_i-z_j)^\ell e^{-B\sum_{i=1}^N|z_i|^2/4}.
\end{equation}
The variables $z_1,\dots z_N$ are complex coordinates of $N$ particles in $\mathbb{R}^2$, which is here identified with the complex plane, and $c_{\mathrm{Lau}}$ is a $L^2$-normalization factor. $\Psi_{\mathrm{Lau}}$ was proposed as a variational trial wave function for the ground state of the many-body Hamiltonian\footnote{Spins are assumed to all be polarized by the magnetic field.}
\begin{equation}
H^{\mathrm{QM}}_N=\sum_{j=1}^N\left[\left(-\im \nabla_{x_j}-\frac{B}{2}x_j^\perp\right)^2+V(x_j)\right]+\lambda\sum_{i<j}W(x_i-x_j)
\end{equation}
acting on the Hilbert space $L^2(\mathbb{R}^{2N})$. Our assumptions on the external (describing trapping and/or the effect of impurities in the sample) potential $V:\mathbb{R}^2\to\mathbb{R}$ and on the interaction $W:\mathbb{R}^2\to\mathbb{R}$ will be specified later, together with the role of the coupling constant $\lambda \ge0$.

The argument leading to \eqref{eq:laughlin} is based on the following two requirements:
\begin{itemize}
\item For a very strong external magnetic field $B$, the leading scale is the magnetic kinetic energy. Hence, it is energetically convenient for each particle to occupy the lowest Landau level (LLL), i.e. the ground eigenspace of the magnetic Laplacian. This allows to restrict the action of $H^\mathrm{QM}_N$ from $L^2(\mathbb{R}^{2N})$ to $\bigotimes^N\mathfrak{H}$ where
\begin{equation} \label{eq:LLL}
\mathfrak{H}=\big\{\psi(z)=f(z)e^{-B|z|^2/4}\in L^2(\mathbb{R}^2)\;|\: f\;\text{ analytic}\big\}.
\end{equation}
\item Once the kinetic energy has been fixed by the reduction to the LLL, one assumes that the next leading scale is given by the interaction $W$. If it is repulsive and strong enough at short range, then it forces the wave function to vanish whenever two particles meet. In order for this prescription to be compatible with \eqref{eq:LLL} the behavior for $z_i\sim z_j$ must be $\sim(z_i-z_j)^\ell$ for some\footnote{The parity of the integer $\ell$ allows to differentiate the statistics of particles, namely, a Laughlin wave function with even $\ell$ describes bosons, while odd values of $\ell$ correspond to fermions.
} integer $\ell$.
\end{itemize}

\medskip

We remark that these prescriptions to derive \eqref{eq:laughlin} actually yield the larger space
\begin{equation} \label{eq:LLL_degeneracy}
\mathcal{L}^N_{\ell,B}=\Big\{\Psi_F=F(z_1,\dots,z_N)\Psi_{\mathrm{Lau}}(z_1,\dots,z_N)\in L^2(\mathbb{R}^{2N})\;|\;F\text{ analytic and symmetric}\Big\}.
\end{equation}
Indeed, any $\Psi_F\in\mathcal{L}^N_{\ell,B}$ belongs to $\bigotimes^N\mathfrak{H}$ and vanishes at least as $(z_i-z_j)^\ell$ when $z_i\sim z_j$. We shall refer to $F$ as a correlation factor. We require it to be symmetric for $\Psi_F$ to have the same symmetry (bosonic or fermionic) under variable exchange as $\Psi_\mathrm{Lau}$.

The energy of a generic $\Psi_F$ is (we subtract the magnetic kinetic energy, which is just a constant within the LLL)
\begin{equation} \label{eq:many_body_energy}
\mathcal{E}_{N,\lambda}[\Psi_F]=\Big\langle\Psi_F\Big|\sum_{j=1}^NV(x_j)+\lambda \sum_{i<j}W(x_i-x_j)\Big|\Psi_F\Big\rangle
\end{equation}
and we define ($B>0$ and $\ell$ a positive integer will be fixed throughout the paper)
\begin{equation}  \label{eq:qm_energy}
E (N,\lambda)=\inf\Big\{\mathcal{E}_{N,\lambda}[\Psi_F]\;|\;\Psi_F\in\mathcal{L}^N_{\ell,B},\,\int_{\mathbb{R}^{2N}}|\Psi_F|^2=1\Big\}.
\end{equation}
The motivation for the above variational problem is as follows: assuming that the two previous requirements are indeed imposed by the energy scales of the problem, what kind of ground state does the system actually choose, within the huge space~\eqref{eq:LLL_degeneracy} ? 

This has a flavor of degenerate perturbation theory, where one looks for the lowest eigenstate of a perturbed Hamiltonian, within the ground eigenspace of the unperturbed part. Here the ``unperturbed part'' should be the magnetic kinetic energy, plus the short-range part of the interaction. We assume these set the main energy scales. It is believed/conjectured that (approximately, for the real model, exactly, for some toy models) the ground state space~\eqref{eq:LLL_degeneracy} for this unperturbed part is protected by a finite energy gap, independent of $N$ (see~\cite{Rougerie-xedp19} for more discussion). We are thus investigating the problem perturbatively, at energies below this (conjectured) energy gap.

In Laughlin's theory, his wave-function~\eqref{eq:laughlin} plays the role of a new ``vacuum''. Crucial to the arguments is the nature of the excitations of this vacuum, and their quasi-particle nature. It is argued that such quasi-particles excitations can be described by wave-functions of the form
\begin{equation} \label{eq:restricted_class}
\Psi_f(z_1,\dots,z_N)=\prod_{j=1}^Nf(z_j)\Psi_{\mathrm{Lau}}(z_1,\dots,z_N)
\end{equation}
where $f$ is a polynomial in a single variable. The function $\Psi_f$ clearly vanishes in correspondence of the zeroes of the polynomial $f$. This is interpreted as having added quasi-holes~\cite{Laughlin-83,Laughlin-87} on top of the Laughlin wave function, the latter being  interpreted as the vacuum state for quasi-holes. 
These quasi-holes carry a fractional charge corresponding to $1/\ell$ of the electron's charge\footnote{Two quasi-holes at the position $z_0$ are interpreted as one quasi-hole with double charge at the position~$z_0$.}, and they are expected to behave as anyons~\cite{AroSchWil-84,LunRou-16} with statistics parameter $-1/\ell$. The aim of this work is to show that the energy~\eqref{eq:qm_energy} can be reached (asymptotically, for large particle numbers) by restricting to the above sub-class, consistently with the physical picture of quasi-holes generation.

We define the minimal energy within this class as
\begin{equation} \label{eq:qh_energy}
e (N,\lambda)=\inf\Big\{\mathcal{E}_{N,\lambda}[\Psi_f]\;|\;\Psi_f \text{ of the form \eqref{eq:restricted_class} },\,\int_{\mathbb{R}^{2N}}|\Psi_f|^2=1\Big\}.
\end{equation}
Notice that, clearly,
\begin{equation} \label{eq:trivial_upper_bound}
E (N,\lambda)\le e (N,\lambda),
\end{equation}
since $e (N,\lambda)$ is the minimal energy in a smaller domain than the one corresponding to $E (N,\lambda)$. We aim at proving that actually
\begin{equation} \label{eq:anticipation}
\boxed{E (N,\lambda)\simeq e (N,\lambda)\quad\text{as}\;N\to\infty \mbox{ with } \lambda \mbox{ fixed}.}
\end{equation}
This means that it is never favorable to add correlations through $F$ in order to minimize the external potential energy and the interaction energy. Furthermore, the lowest energy is attainable with uncorrelated quasi-holes exciting the Laughlin wave function. 

In previous papers~\cite{LieRouYng-16,LieRouYng-17,RouYng-17}, the case $\lambda = 0$ was considered, i.e. it was assumed that restricting the variational set to~\eqref{eq:LLL_degeneracy} rendered the interaction energy negligible compared to the potential energy. There~\eqref{eq:anticipation} was proved for external potentials varying on the characteristic scale of the Laughlin function (a technical assumption that we shall also make, see below for further discussion). 

In the present contribution we consider the effect of $W$, and prove~\eqref{eq:anticipation} for suitably scaled potentials $V$ and $W$ and a sufficiently small coupling constant $\lambda$. This means that the Laughlin phase (perturbation of the Laughlin function by uncorrelated quasi-holes, as in~\eqref{eq:restricted_class}) is stable against both external potentials and the long-range part of the interaction (the short-range part is supposed to have been taken into account by restricting to~\eqref{eq:LLL_degeneracy}).
 
\bigskip 

\noindent \textbf{Acknowledgments:} We thank Jakob Yngvason for useful discussions. Funding from the European Research Council (ERC) under the European Union's Horizon 2020 Research and Innovation Programme (Grant agreement CORFRONMAT No 758620) is gratefully acknowledged.

%
%

%
%

\section{Setting and theorems} \label{sect:setting}

\subsection{Main results}

We now present the precise setting in which we prove~\eqref{eq:anticipation} and discuss further this main result and its corollaries.


For any $\Psi_F\in\mathcal{L}_{\ell,B}^N$ let us introduce the associated reduced densities
\begin{equation}
\rho^{(k)}_F(z_1,\dots,z_k)=\binom{N}{k}\int_{\mathbb{R}^{2(N-k)}}|\Psi_F(z_1,\dots,z_k,z_{k+1},\dots,z_N)|^2dz_{k+1}\dots dz_N
\end{equation}
for $k=1,\dots,N$. Notice that we are choosing the convention 
$$\int_{\R^{2k}}\rho^{(k)}_F=\binom{N}{k}.$$
It was proven in \cite{LieRouYng-17} (see also \cite{RouYng-14,RouYng-15} for earlier partial results) that the Laughlin fluid satisfies an incompressibility estimate. More precisely, for any $\Psi_F\in\mathcal{L}^N_{\ell,B}$, the associated one-particle density satisfies
\begin{equation} \label{eq:incompressibility}
\rho^{(1)}_F(z)\lesssim\frac{B}{2\pi\ell}
\end{equation}
in a suitable sense of local averages. This is a rather nontrivial rigidity property: no matter how complicated the correlation factor $F$ might be, it cannot locally compress the liquid's density beyond a universal threshold. The proof of~\eqref{eq:incompressibility} relies on the \emph{plasma analogy}, a useful mapping of the problem known since Laughlin's work \cite{Laughlin-83,Laughlin-87}: the many-body density of $\Psi_F$ can be mapped to a classical Hamiltonian for $N$ particles in 2D. Within this framework, a generic $F$, due to its analyticity, plays the role of the potential generated by a \emph{positive} charge distribution (see \cite[Section 2.3]{LieRouYng-17}). Since the $N$ classical particles are also (by convention) positively charged in the plasma analogy, one can expect that the overall effect of $F$ is to \emph{decrease} the density $\rho^{(1)}_F$ with respect to the case $F=1$. The analogue of \eqref{eq:incompressibility} for this latter case had already been argued~\cite{Laughlin-83}, and then rigorously derived~\cite{RouSerYng-13a,RouSerYng-13b}.

We remark that the Laughlin phase naturally occupies a length scale of order $N^{1/2}$, as can easily be seen by combining $\int\rho^{(1)}_F=N$ with \eqref{eq:incompressibility}. We then define new densities with rescaled lengths according to
\begin{equation} \label{eq:rescaled_density}
\mu_F^{(k)}(x_1,\dots,x_k)=\frac{N^k}{\binom{N}{k}}\rho^{(k)}_F\left(\sqrt{N}x_1,\dots,\sqrt{N}x_k\right).
\end{equation}
Note that we set 
$$ \int_{\R ^{2k}} \mu_F ^{(k)} = 1.$$
We also introduce rescaled potentials $v$ and $w$ through
\begin{equation}
v(x)=V\left(\sqrt{N}x\right)
\end{equation}
and
\begin{equation} \label{eq:rescaled_w}
\frac{w(x)}{N}=W\left(\sqrt{N}x\right),\qquad\text{with }w(-x)=w(x).
\end{equation}
From now on we will only work in the scaled variables $x=z/\sqrt{N}$. The notation we introduced allows to bring \eqref{eq:many_body_energy} to the form
\begin{equation} \label{eq:qm_functional}
\mathcal{E}_{N,\lambda}[\Psi_F]=N\left[\int_{\mathbb{R}^2}v(x)\mu^{(1)}_F(x)dx+\frac{\lambda}{2}\iint_{\mathbb{R}^2\times\mathbb{R}^2}w(x-y)\mu^{(2)}_F(x,y)dxdy\right].
\end{equation}
The definition of $w$ in \eqref{eq:rescaled_w} effectively puts us in the mean-field regime, by making the two contributions to \eqref{eq:qm_functional} formally of the same order. Recall that the reduced many-body ground state energy is
\begin{equation*}
E (N,\lambda)=\inf\Big\{\mathcal{E}_{N,\lambda}[\Psi_F]\;|\;\Psi_F\in\mathcal{L}^N_{\ell,B},\,\int_{\mathbb{R}^{2N}}|\Psi_F|^2=1\Big\}
\end{equation*}
and the energy within states with uncorrelated quasi-holes is
\begin{equation*}
e (N,\lambda)=\inf\Big\{\mathcal{E}_{N,\lambda}[\Psi_f]\;|\;\Psi_f \text{ of the form \eqref{eq:restricted_class} },\,\int_{\mathbb{R}^{2N}}|\Psi_f|^2=1\Big\}.
\end{equation*}
As customary in the study of asymptotic properties of many-body systems, a crucial role is played by the associated effective mean-field model. Let us define the functional
\begin{equation} \label{eq:mf_functional}
\mathcal{E}^\mathrm{MF}[\mu]=\int_{\mathbb{R}^2}v(x)\mu(x)dx+\frac{\lambda}{2}\iint_{\mathbb{R}^2\times\mathbb{R}^2}w(x-y)\mu(x)\mu(y)dxdy,
\end{equation}
for any probability density $\mu$ such that the integrals make sense. This is the mean-field version of \eqref{eq:qm_functional} (having ignored the $N$ multiplicative factor), in that the two-body density has been replaced by the uncorrelated product of one-body densities.

We will show that the minimal energy attained by $\mathcal{E}^\mathrm{MF}$ is related to both $E (N,\lambda)$ and $e (N,\lambda)$. However, the mere minimization of $\mathcal{E}^{\mathrm{MF}}$ under the sole mass constraint $\int\mu=1$ would be insensitive of the fact that variational states for $E (N,\lambda)$ and $e (N,\lambda)$ satisfy the incompressibility bound \eqref{eq:incompressibility}. For this reason, we define the `flocking' energy as
\begin{equation} \label{eq:mf_energy}
E^{\mathrm{flo}}=\inf\Big\{\mathcal{E}^{\mathrm{MF}}[\mu]\;|\;0\le\mu\le\frac{B}{2\pi\ell},\,\int_{\mathbb{R}^2}\mu=1\Big\}.
\end{equation}
The upper constraint is precisely the one in \eqref{eq:incompressibility}. 
We use the word `flocking', since \eqref{eq:mf_energy} is the minimal energy for a functional where confining and repulsive terms compete, and there is an overall constraint on the maximal density. This mechanism resembles the one occurring in models for bird flocks.

Minimization problems similar to \eqref{eq:mf_energy} have been studied in \cite{BurChoTop-15} and \cite{FraLie-16}, although within a slightly different setting, namely when $v=|\cdot|^p*\mu$ for some $p>0$, hence with a dependence on $\mu$. Depending on the parameters of the problem, the behavior of a minimizer $\mu_m$ can exhibit three phases:
\begin{itemize}
\item `Liquid phase': $|\{\mu_m=B/2\pi\ell\}|=0$,
\item `Intermediate phase': $|\{0<\mu_m<B/2\pi\ell\}|>0$ and $|\{\mu_m=B/2\pi\ell\}|>0$,
\item `Solid phase': $|\{0<\mu_m<B/2\pi\ell\}|=0$.
\end{itemize}
It was proven in \cite{BurChoTop-15} and \cite{FraLie-16} that, for $v=|\cdot|^p*\mu$ and for suitable values of $p$, the value of the mass constraint $\int\mu=M$ discriminates between the phases: for small enough $M$ any minimizer is in the liquid phase, while for large enough $M$ any minimizer is in the solid phase. In our setting, the parameter which governs the behavior of minimizers is the coupling constant. The flocking energy~\eqref{eq:mf_energy} with $\lambda=0$ coincides with the classical \emph{bathtub} energy~\cite[Theorem 1.14]{LieLos-01}, namely the minimal possible energy in an external potential achievable by densities of fixed mass with upper and lower constraints. It is well-known that for this problem the minimizers are always in the solid phase, and saturate the upper constraint almost everywhere on their support. 


%

We now present our assumptions on the potentials $v$ and $w$.

\begin{assumption}[\textbf{The external potential}]\mbox{}\label{assum:external}\\ 
There exists a fixed function $v\in C^2(\mathbb{R}^2,\mathbb{R}^+)$ such that
\begin{equation*}
V(z)=v\Big(\frac{z}{\sqrt{N}}\Big).
\end{equation*}
We assume that $v$ has a finite number of critical points, all of which are non-degenerate. We also assume that $v$ has polynomial growth, i.e.,
\begin{equation}\label{eq:growth_trap}
|x|^s\le v(x)\le |x|^t\qquad\text{for }|x|\text{ large enough}
\end{equation}
for some $t\ge s\ge1$.
\end{assumption}

\begin{assumption}[\textbf{The interaction potential}]\mbox{}\label{assum:interaction}\\ 
There exists a fixed even function $w\in W^{2,\infty}(\mathbb{R}^2,\mathbb{R})$ (two bounded derivatives) with 
$$ w(x) \underset{|x|\to \infty} \to 0$$
such that
\begin{equation*}
W(z)=\frac{1}{N}w\left(\frac{z}{\sqrt{N}}\right).
\end{equation*}
\end{assumption}

Our main result is stated next. Note that we do not assume $w$, nor the coupling constant $\lambda$, to have a sign. The case most relevant to FQH physics has $\lambda w \geq 0$ but our proofs allow for attractive potentials as well. 

\begin{theorem}[\textbf{Energy of the Laughlin phase}]\mbox{}\label{thm:main}\\ 
For $V$ and $W$ satisfying the assumptions presented above there exists $\lambda_0>0$ such that, for $|\lambda| \le \lambda_0$,
\begin{equation} \label{eq:energy_convergence}
\lim_{N\to\infty}\frac{E (N,\lambda)}{e (N,\lambda)}=1.
\end{equation}
\end{theorem}

From the proof we can deduce properties of (quasi-)minimizers: 

\begin{theorem}[\textbf{Convergence of densities}]\mbox{} \label{thm:densities}\\
Let $F$ be a (sequence of) correlation factors such that the associated $\Psi_F\in\mathcal{L}_{\ell,B}^N$ satisfy
\begin{equation} \label{eq:minimizing_sequence}
\mathcal{E}_{N,\lambda}[\Psi_F]=E (N,\lambda)+o(N)
\end{equation}
as $N\to\infty$. Assume $|\lambda| \le \lambda_0$ as in Theorem \ref{thm:main}. Then there exists a probability measure $P$ supported on the set of minimizers of the flocking problem \eqref{eq:mf_energy} such that, for every $k\in\mathbb{N}$,
\begin{equation*}
\mu_F^{(k)}\rightharpoonup\int_{\{\mu\;|\;\mathcal{E}^\mathrm{MF}[\mu]=E^\mathrm{flo}\}}\mu^{\otimes k}dP(\mu)
\end{equation*}
weakly as probability measures as $N\to\infty$.
\end{theorem}

\subsection{Remarks}

As explained before, the scaling we chose in~\eqref{eq:rescaled_w} is motivated by the fact that it renders potential and interaction energies of the same order of magnitude, so that Theorem~\ref{thm:main} is not a perturbative statement. On the other hand, the regularity assumptions on the data of the problem are mostly of a technical nature. In this section we further discuss these aspects.

\medskip

\noindent\textbf{Scaling of the problem.} Although our Assumptions~\ref{assum:external}-\ref{assum:interaction} do not cover it (see below), it is instructive to think of the original, unscaled potentials as 
$$ W (x) = \frac{1}{|x|}, \quad V(x) = - \rho_{\rm ext} \star |\,.\, |^{-1}$$
i.e. the interactions are electrostatic\footnote{Recall that we consider 3D electrons confined to 2D, so the Coulomb potential should be the 3D one.} and the external potential is generated by a fixed density of charge $-\rho_{\rm ext}$. Since the system lives on a thermodynamic length scale $\propto \sqrt{N}$ it is natural to scale $\rho_{\rm ext}$ as 
$$ \rho_{\rm ext} = \mu_{\rm ext} ( N^{-1/2} x  ).$$
A scaling of length units, using~\eqref{eq:rescaled_density}, then gives, for these choices, 
\begin{equation}\label{eq:coulomb case}
\mathcal{E}_{N,\lambda}[\Psi_F]=N^{3/2}\iint_{\mathbb{R}^2 \times \mathbb{R}^2}\frac{1}{|x-y|}\left[-\mu_\mathrm{ext}(x)\mu^{(1)}_F(y)+\frac{\lambda}{2}\mu^{(2)}_F(x,y)\right]dxdy. 
\end{equation}
The scaling~\eqref{eq:qm_functional} used for our main theorem mimics the above, where the two terms are also of the same order. 

\medskip

\noindent\textbf{Link with filling factor.} In~\eqref{eq:coulomb case}, increasing $\lambda$ is equivalent to decreasing $\mu_{\rm ext}$. Neutrality (or almost-neutrality) considerations for the full system ($\rho ^{\rm ext}$ minus the electrons' charge) make it natural to expect (or demand) that the latter be in relation with the electron density, so that increasing $\lambda$ (decreasing $\mu_{\rm ext}$) is related to decreasing the filling factor. From this point of view it is important that the system stays of the form ``Laughlin plus quasi-holes'' for $|\lambda| \leq \lambda_0$. This is one of the ingredients of Laughlin's explanation of the quantization of the Hall conductivity for filling factors close to $1/n$ (and not just equal to $1/n$). This is indeed what we prove, in a simplified model at least. It is also natural to expect the restriction to small (but independent of $N$) values of $\lambda$ to be necessary: for larger values/lower filling factors, the system should form another FQH ground state (e.g. a Laughlin state with higher exponent), corresponding to another plateau of the Hall conductivity.

Indeed, the variational set~\eqref{eq:LLL_degeneracy} contains states corresponding to filling factors lower than $\ell^{-1}$, for example $\Psi_{\mathrm{Lau}}$ with exponent $\ell + 2$ which is $\Psi_{F_2}$ with
\begin{equation*}
F_2(z_1,\dots,z_N)=\prod_{i\le j}(z_i-z_j)^2.
\end{equation*}
This state intuitively has a smaller interaction energy (since it vanishes faster than $(z_i-z_j)^\ell$) than quasi-holes states. Our result says that, for small coupling constants, it is favorable to generate quasi-holes without changing the $\ell$ exponent. The fact that, for small $\lambda$, increasing $\ell$ increases the total energy can be explicitly seen in the following simple case. Assume that $v$ and $w$ are positive and radially symmetric, and that $\widehat w\ge0$. Then one can prove that the flocking problem \eqref{eq:mf_energy} has the unique minimizer
\begin{equation*}
\mu_\mathrm{min}=\frac{B}{2\pi\ell}\mathbbm{1}_{D(0,R)}
\end{equation*}
where the radius $R$ of the disk is fixed by normalization. On the other hand, the (normalized) wave-function $\Psi_{F_2}$ coincides with the Laughlin wave-function with filling factor $(\ell+2)^{-1}$, and therefore its density $\mu_{F_2}^{(1)}$ is known (see, e.g., \cite[Theorem 3.1 and Proposition 3.1]{RouSerYng-13b}) to satisfy
\begin{equation*}
\mu_{F_2}^{(1)}\rightharpoonup\frac{B}{2\pi(\ell+2)}\mathbbm{1}_{D(0,R')}
\end{equation*}
in the weak sense of probability measures. The two above equations mean that the Laughlin state with exponent $\ell$ ($F\equiv1$) gives the best energy, while the Laughlin state $\Psi_{F_2}$ does not.

\medskip

\noindent\textbf{Singular interactions.} The main restrictive assumption we make is the smoothness of $w$ in Assumption~\ref{assum:interaction}. This is certainly not satisfied in the Coulomb case we just described. 
There are a few reasons why our main results are nevertheless relevant for 2D electron gases:
\begin{itemize}
 \item The finite (although small) thickness of the gas in the direction perpendicular to the plane yields a smoothing of the Coulomb interaction's singularity at short distances, see e.g. the discussion in~\cite[Section~2.2]{Jain-07}. 
 \item It is likely that the restriction to the lowest Landau level (a fortiori the emergence of Laughlin's function) is valid only on length-scales large compared to the interparticle distance. At short distances, specific correlations could be formed to avoid the singular part of the interaction. See~\cite{LewSei-09,SeiYng-20} for results in this direction in a related situation.
 \item The Jastrow factor in~\eqref{eq:laughlin}, i.e. the product $\prod_{i<j}(z_i-z_j)^\ell$, should kill the singularity of the true Coulomb interaction, allowing to replace it by an effective, smoother, potential $W$.
\end{itemize}
In any event, as indicated in our title, our main concern is the stability of the Laughlin phase against the long-range part of the interaction, which is certainly present in the coulombic one, but neglected in the derivation of the ansatz. See~\cite{LieSolYng-95} for mathematical results on 2D electron gases with the true Coulomb interaction. 

\medskip

\noindent\textbf{Perturbations on mesoscopic length scales.} For simplicity we consider a model with potentials scaled in such a way that they affect the shape of the electron droplet on macroscopic length scales, comparable to the size of the full system. Physically it would be relevant to allow for modifications on much smaller length scales (with the limitation that they should still be much larger than the typical interparticle distance). See~\cite[Section~5]{RouYng-17} for comments in this direction. The main bottleneck for a mathematical analysis of such situations is to improve the incompressibility estimates of~\cite{LieRouYng-17}. For the pure Laughlin state $F=1$, the results of~\cite{Leble-15b,LebSer-16,BauBouNikYau-15,BauBouNikYau-16} go in this direction.

\subsection{Sketch of proof}

Let us now list the main ingredients of the proofs of Theorem~\ref{thm:main} and Theorem~\ref{thm:densities}. There are three main steps: 
\begin{itemize}
 \item The incompressibility estimate~\eqref{eq:incompressibility}, which holds for any correlation factor $F$, suggests that the flocking energy~\eqref{eq:mf_energy} is a lower bound to the full energy $E (N,\lambda)$.  
 \item For small enough $\lambda$ we prove that the flocking minimizers are in the solid phase, i.e. their density takes only the values $0$ or $B/(2\pi\ell)$.
 \item For any solid flocking minimizer\footnote{Actually, any function taking only the values $0$ and $B/(2\pi \ell)$} $\mu_m$ we can construct a recovery sequence of ``Laughlin plus quasi-holes'' wave functions of the form~\eqref{eq:restricted_class} whose (rescaled) density converges to $\mu_m$.
\end{itemize}

The first point shows that $E (N,\lambda) \gtrapprox N \Eflo$. The second and third taken together show that, for $\lambda$ small enough, $e (N,\lambda) \lessapprox N \Eflo$. Since obviously $E (N,\lambda) \leq e (N,\lambda)$ this gives the scheme of proof for Theorem~\ref{thm:main}. Theorem~\ref{thm:densities} requires some more tools, in particular the classical de Finetti-Hewitt-Savage theorem~\cite{HewSav-55}, and in fact its constructive proof due to Diaconis and Freedman~\cite{DiaFre-80}.  

The first point is the most important. In fact, we need not only justify~\eqref{eq:incompressibility} in a rigorous way, but also understand why a mean-field approximation is valid for a lower bound, leading to the flocking energy~\eqref{eq:mf_energy}.  The key is to prove that the empirical density of a state of the form~\eqref{eq:LLL_degeneracy} satisfies~\eqref{eq:incompressibility} with large probability. This is an improvement over the incompressibility estimates derived in~\cite{LieRouYng-16,LieRouYng-17}. The latter only gave a bound in average, as stated informally in~\eqref{eq:incompressibility}. We supplement this information with deviations estimate in Theorem~\ref{thm:exponential_bound} below and deduce the 

\begin{theorem}[\textbf{Lower bound to the many-body energy}]\mbox{}\label{thm:energies low} \\ 
 Suppose Assumptions \ref{assum:external} and \ref{assum:interaction} hold. For any $\Psi_F$, we have
\begin{equation}\label{eq:lower_bound}
\begin{split}
\mathcal{E}_{N,\lambda}[\Psi_F]\ge NE^{\mathrm{flo}}\big(1-CN^{(-3+\sqrt{5})/4+\gamma}\big)
\end{split}
\end{equation}
for any $\gamma>0$.
\end{theorem}

The second point of the overall strategy, concerning the flocking minimizers, is reminiscent of results from~\cite{BurChoTop-15,FraLie-16}. We shall prove the following: 

\begin{theorem}[\textbf{Solid phase of the flocking problem}]\mbox{}\label{thm:minimizer} \\ 
Suppose Assumptions \ref{assum:external} and \ref{assum:interaction} hold. Then:
\begin{itemize}
\item[$(i)$]  There exists a minimizer for the problem \eqref{eq:mf_energy}.
\item[$(ii)$] There exists $\lambda_0>0$ such that, for any $|\lambda| \leq \lambda_0$, any minimizer $\mu_m$ is in the solid phase, i.e. 
\begin{equation*}
\mu_m=\frac{B}{2\pi\ell}\mathbbm{1}_\Sigma\qquad \text{a.e.}
\end{equation*}
for some open set $\Sigma\subset\mathbb{R}^2$.
\end{itemize}
\end{theorem}

Theorem~\ref{thm:minimizer} is proven in Section \ref{sect:flocking}. The main idea is that on any patch of liquid phase $0 < \mu_m < B / (2\pi \ell)$, the external potential and the mean-field potential must for optimality add up to a constant. But for small enough $\lambda$ the variations of the mean-field potential cannot possibly cancel those of the external potential, because of the flocking constraint.  

\medskip

The third point of the overall strategy is

\begin{theorem}[\textbf{Upper bound to the many-body energy}]\mbox{} \label{thm:energies up}\\
Let $\mu^{\mathrm{sol}}$ be any fixed probability measure on $\mathbb{R}^2$ in the solid phase, i.e., whose only values are zero and $B/2\pi\ell$. Then there exists (a sequence of) polynomial(s) $f_\delta$, indexed by a positive parameter $\delta\underset{N\to\infty}{\to} 0$, such that
\begin{equation}  \label{eq:upper_bound}
\begin{split}
e (N,\lambda)\le \mathcal{E}_{N,\lambda}&[\Psi_{f_\delta}]\le N\mathcal{E}^\mathrm{MF}[\mu^\mathrm{sol}]\left(1+CN^{-1/4+\gamma}\right).
\end{split}
\end{equation}
for any $\gamma>0$.
\end{theorem}

Basically the above follows by adapting the methods of~\cite{RouYng-17}, but we have to handle in addition the two-body term of the energy, absent in the aforementioned reference. We remark that requiring $\mu^\mathrm{sol}$ to be in the solid phase is crucial. Indeed, essentially any trial state of the form~\eqref{eq:restricted_class} constructed from uncorrelated quasi-holes leads to a density of such a form.

\medskip

We will prove~\eqref{eq:lower_bound} in Section~\ref{sect:lower_bound} and \eqref{eq:upper_bound} in Section \ref{sect:upper_bound}. We explain here how combining Theorems~\ref{thm:energies low}-\ref{thm:minimizer}-\ref{thm:energies up} yields the proof of Theorem \ref{thm:main}. Theorem \ref{thm:densities} requires other ingredients, provided in Section \ref{sect:densities}.

\begin{proof}[Proof of Theorem \ref{thm:main}]
First, recall that $E (N,\lambda)\le e (N,\lambda)$ as already stated in \eqref{eq:trivial_upper_bound}. Moreover, under the assumptions of Theorem~\ref{thm:main},  there exists by Theorem \ref{thm:minimizer} a minimizer $\mu_m$ of the flocking problem in the solid phase. Hence, we are allowed to choose $\mu^\mathrm{sol}=\mu_m$ in \eqref{eq:upper_bound}, and we have
\begin{equation*}
\mathcal{E}^{\mathrm{MF}}[\mu^\mathrm{sol}]=\mathcal{E}^\mathrm{MF}[\mu_m]=E^\mathrm{flo}
\end{equation*}
By combining \eqref{eq:lower_bound} and \eqref{eq:upper_bound} we deduce
\begin{equation} \label{eq:proof_energies}
NE^\mathrm{flo}\big(1-CN^{(-3+\sqrt{5})/4+\gamma}\big)\le  E (N,\lambda)\le e (N,\lambda)\le NE^\mathrm{flo}\big(1+CN^{-1/4+\gamma}\big),
\end{equation}
which proves Theorem \ref{thm:main}.
\end{proof}

\section{Solid-phase minimizers of the flocking problem} \label{sect:flocking}

This section is devoted to the proof of Theorem \ref{thm:minimizer}. Recall that we are considering the functional
\begin{equation*}
\mathcal{E}^\mathrm{MF}[\mu]=\int_{\mathbb{R}^2}v(x)\mu(x)dx+\frac{\lambda}{2}\iint_{\mathbb{R}^2\times\mathbb{R}^2}w(x-y)\mu(x)\mu(y)dxdy
\end{equation*}
on the variational set
\begin{equation*}
\mathcal{M}=\left\{ \mu\in L^1(\mathbb{R}^2)\cap L^\infty(\mathbb{R}^2)\;|\;0\le\mu\le\frac{B}{2\pi\ell}\right\}
\end{equation*}
and the minimization problem
\begin{equation}\label{eq:Eflo again} 
E^{\mathrm{flo}}=\inf\left\{\mathcal{E}^{\mathrm{MF}}[\mu]\;|\;\mu\in\mathcal{M}, \,\int_{\mathbb{R}^2}\mu=1 \right\}.
\end{equation}
A priori, the set $L^1(\mathbb{R}^2)$ is not closed with respect to the weak topology, and this would not allow to deduce that minimizing sequences have a weak limit (up to a subsequence). On the other hand, we can see the above minimization problem as defined on probability measures $\mu$ such that
\begin{equation*}
\frac{B}{2\pi\ell}dx-\mu\ge0,\,\text{ and }\,\int_{\mathbb{R}^2}\mu=1.
\end{equation*}
The first condition actually implies that every such $\mu$ is absolutely continuous with respect to the Lebesgue measure, and that its density $\mu(x)$ satisfies $0\le\mu(x)\le B/2\pi\ell$ on $\mathbb{R}^2$. This implies that the set $\mathcal{M}$, as a subset of $L^1(\mathbb{R}^2)$, is closed under the weak topology.

\subsection{Existence of minimizers}\label{sec:exist}

This is fairly straightforward. See e.g.~\cite{SafTot-97,Serfaty-15} for background and~\cite{BurChoTop-15,FraLie-16,LieRouYng-17} for similar problems. Since $v$ and $w$ are continuous and bounded below, a monotone convergence argument shows that the map $\mu\mapsto \mathcal{E}^\mathrm{MF}[\mu]$ is lower semicontinuous. Since $v(x) \underset{|x|\to \infty} \to +\infty$, any minimizing sequence must be tight as a sequence of probability measures, and converge to an element of $\mathcal{M}$ with mass $1$. By lower semicontinuity this gives the existence of a minimizer for~\eqref{eq:Eflo again}.
%
%
%

\subsection{Solid phase}

We now prove that, for small enough coupling constants, the minimizers are in the solid phase. For any $\mu\in\mathcal{M}$, define the total potential
\begin{equation}
\Phi_\mu=v+\lambda w*\mu.
\end{equation}
Clearly, by Young's inequality, Assumption~\ref{assum:interaction} implies the 

\begin{lemma}[\textbf{Bounds on the mean-field potential}]\label{lem:pot}\mbox{}\\
For $w\in W^{2,\infty}(\mathbb{R}^2)$ and $\mu\in L^1(\mathbb{R}^2)$ we have
\begin{equation}
\begin{split}  \label{eq:bounded_convolution}
\|w*\mu\|_\infty\le\;& C_w\|\mu\|_{L^1}\\
\|\nabla w*\mu\|_\infty\le\;& C_w\|\mu\|_{L^1}\\
\left\| \partial_{\yv}^2 w*\mu\right\|_\infty\le\;& C_w\|\mu\|_{L^1}
\end{split}
\end{equation}
for any unit vector $\yv$ and for a constant $C_w>0$ that depends only on $w$.
\end{lemma}

\begin{lemma}[\textbf{Variational inequalities}]\mbox{}\label{lemma:inequalities_external}\\ 
	Let $\mu$ be a local minimizer of $\mathcal{E}^{\mathrm{MF}}_{}$ on $\mathcal{M}$. Then there exists $\gamma\in\mathbb{R}$ such that, for almost every $x$,
	\begin{equation} \label{eq:variational_external}
	\begin{split}
	\Phi_{\mu}(x)\le \gamma&\quad\text{if } \mu(x)=\frac{B}{2\pi\ell}\\
	\Phi_{\mu}(x)= \gamma&\quad\text{if } 0<\mu(x)<\frac{B}{2\pi\ell}\\
	\Phi_{\mu}(x)\ge \gamma&\quad\text{if } \mu(x)=0.
	\end{split}
	\end{equation}
\end{lemma}

\begin{proof}
The proof is a straightforward adaptation of~\cite[Lemma 4.2]{BurChoTop-15}. The variational inequalities are obtained in the standard way by testing the energy of $\mu_\eps := \mu + \eps \varphi$ with $\eps >0$ a small number and $\varphi$ a (say smooth) function satisfying
$$ \begin{cases}
    \varphi \geq 0 \mbox{ where } \mu = 0\\
    \varphi \leq 0 \mbox{ where } \mu = \frac{B}{2\pi\ell}\\
    \int_{\R^2} \varphi = 0.
   \end{cases}
$$
To lower order in $\eps$ the condition for $\mu_\eps$ to have a larger energy than $\mu$ is
$$ \int_{\R^2} \left( v + \lambda w * \mu \right) \varphi \geq 0.$$
This being so for any $\varphi$ as above proves the claimed variational inequalities.
\end{proof}

We now complete the proof of Theorem \ref{thm:minimizer}.
\begin{proof}[Proof of Theorem \ref{thm:minimizer}, Item $(ii)$] Pick a global minimizer $\mu_\mathrm{m}$ and define
	\begin{equation*}
	A:=\Big\{x\in\mathbb{R}^2\,|\,0<\mu_\mathrm{m}(x)<\frac{B}{2\pi\ell}\Big\}.
	\end{equation*}
	We will show that $A$ has measure zero for a range of coupling constants. Assume first 		that $A$ does not accumulate to any critical point of $v$, i.e., there exists $c_1$ such that
	\begin{equation*}
	|x-x_k|> c_1>0\quad\forall x \in A,
	\end{equation*}	
	where $(x_k)_k$ is the \emph{finite} set of critical points of $v$. This implies the existence of $c_2$ such that
	\begin{equation*}
	|\nabla v(x)|> c_2>0\quad\forall x \in A.
	\end{equation*}
	Let us assume that $A$ is an open set. Then, since by \eqref{eq:variational_external} $\Phi_{\mu_\mathrm{m}}$ is constant on $A$, we can differentiate and obtain
	\begin{equation} \label{eq:equality_on_A}
	\nabla v+\lambda \nabla w*\mu_\mathrm{m}=0 \qquad\text{ on }A.
	\end{equation}
	By taking $\lambda_0'=c_2/\|\nabla w*\mu_m\|_\infty$ we see that last two formulae are in contradiction for any $|\lambda| \leq\lambda_0'$ due to
\begin{equation*}
c_2<|\nabla v(x)|=|\lambda| |\nabla w*\mu_m|\le\lambda_0'|\nabla w*\mu_m|\le \frac{c_2}{\|\nabla w*\mu_m\|_\infty}|\nabla w*\mu_m|\le c_2.
\end{equation*}	
	 This excludes the possibility of $A$ being open and non-empty, and more generally of $A$ containing an open set.
	
	We still have to rule out the case in which $A$ has positive measure without containing open sets. However, if a function belongs to $W^{1,1}_{\mathrm{loc}}$, then its gradient vanishes almost everywhere on the pre-image of any point (see, e.g., \cite[Theorem 6.19]{LieLos-01} or \cite[Proposition 4]{FraLie-16}). It is easy to see that $\Phi_{\mu_\mathrm{m}}$ belongs to $W^{1,1}_\mathrm{loc}$ since $v\in C^2$ and the term containing $w$ is controlled in $W^{1,\infty}\subset W^{1,1,}_\mathrm{loc}$ using \eqref{eq:bounded_convolution}. Hence we have that \eqref{eq:equality_on_A} still holds almost everywhere on $A$, and we encounter the same contradiction as above for $|\lambda| \leq\lambda_0'$.
	
	Consider now the case where $A$ might accumulate to  critical points $\{x_1,\dots,x_M\}$ of $v$. The latter are non-degenerate by Assumption \ref{assum:external}. Then, for $k\in\{1,\dots,N\}$, at the critical point $x_k$ there exists a direction $\yv_k$ along which the second derivative of $v$ at $x_k$ is different from zero. Since $v\in C^2$, this is true in a neighborhood of each $x_k$, i.e., there exist $c_{3,k},c_{4,k}>0$ such that
\begin{equation} \label{eq:non_degeneracy}
\Big|\frac{\partial^2}{\partial \yv_k^2}v(x)_{| x=x_k}\Big|>c_{3,k},\quad\forall x\in D (x_k,c_{4,k})
\end{equation}
for any $k=1,\dots,M$, where $D(a,R)$ is the disk of center $a$ and radius $R$.

Define the set
\begin{equation*}
A'=A\setminus\bigcup_{k=1}^M \left( D (x_k,c_{4,k}) \cap A\right),
\end{equation*}
i.e. the set $A$ from which we removed a neighborhood of each critical point of $v$ to which $A$ accumulates. On $A'$,
\begin{equation*}
|x-x_k|> c_4:=\min_k c_{4,k}>0,
\end{equation*}
which implies $|\nabla v(x)|>c_5>0$ for some $c_5$. We can then replicate the above argument and deduce that $A'$ has zero measure for $|\lambda| \leq\lambda_0''$ with $\lambda_0''=c_5/\|\nabla w*\mu_m\|_\infty$.

On $D_{x_k}(c_{4,k})\cap A$, in turn, we can differentiate \eqref{eq:variational_external} twice along the $\yv_k$ direction, thus obtaining
\begin{equation} \label{eq:second_differentiation}
\frac{\partial^2}{\partial \yv_k^2}v(x)+\lambda \frac{\partial^2}{\partial \yv_k^2}w*\mu_m=0.
\end{equation}
This formula is valid almost everywhere even if $D (x_k,c_{4,k}) \cap A$ does not contain an open set since, as above, it only requires $\Phi_{\mu_m}\in W^{2,1}_\mathrm{loc}$ (see again \cite[Proposition 4]{FraLie-16}). Notice that $\partial_{\yv_k}^2\Phi_{\mu_m}\in L^1_\mathrm{loc}$ follows because $v\in C^2$ and $\partial_{\yv_k}^2w*\mu_m\in L^\infty\subset L^{1}_\mathrm{loc}$ due to \eqref{eq:bounded_convolution}. By choosing
\begin{equation*}
\lambda_{0,k}'''=\frac{c_{3,k}}{\Big\|\dfrac{\partial^2}{\partial\yv_k^2} w*\mu_m\Big\|_\infty}
\end{equation*}
we see that \eqref{eq:non_degeneracy} and \eqref{eq:second_differentiation} are in contradiction for $|\lambda| \leq\lambda_{0,k}'''$ and hence $B_{x_k}(c_{4,k})\cap A$ must have zero measure.
The proof is concluded by taking
\begin{equation*}
\lambda_0=\min\big\{\lambda_0',\lambda_0'',\lambda_{0,1}''',\dots,\lambda_{0,M}'''\big \},
\end{equation*}
because, by what we discussed, $A$ must have zero measure for $|\lambda|\leq \lambda_0$. Notice that $\lambda_{0}',\lambda_0'',\lambda_{0,k}'''>0$, whence $\lambda_0>0$.
\end{proof}

\subsection{Slightly perturbed problems}
%
%
%

In Section \ref{sect:lower_bound} we will need to consider the minimization of $\mathcal{E}^\mathrm{MF}$ with a perturbed upper constraint. We now prove a result of continuous dependence on such perturbations. Define
\begin{equation*}
\mathcal{M}_\varepsilon=\Big\{ \mu\in L^1(\mathbb{R}^2)\cap L^\infty(\mathbb{R}^2)\;|\;0\le\mu\le(1+\varepsilon)\frac{B}{2\pi\ell},\,\int_{\mathbb{R}^2}\mu=1\Big\}
\end{equation*}
and
\begin{equation}  \label{eq:energy_perturbed_constraint}
E^{\mathrm{flo}}_{\varepsilon}=\inf\big\{\mathcal{E}^{\mathrm{MF}}[\mu]\;|\;\mu\in\mathcal{M}_\varepsilon\big\}.
\end{equation}

\begin{lemma}[\textbf{Dependence on the flocking constraint}]\mbox{}\label{lemma:dependence_upper_constraint}\\ 
	There exists a constant $C>0$ such that for any $\varepsilon > 0$ small enough we have
	\begin{equation*}
	\begin{split}
	E^\mathrm{flo}_{\varepsilon}\le E^\mathrm{flo}\le E^\mathrm{flo}_{\varepsilon}+C\varepsilon
	\end{split}
	\end{equation*}
\end{lemma}

\begin{proof}
	The inequality
	\begin{equation*}
	E_{\varepsilon}^{\mathrm{flo}}\le E^\mathrm{flo}
	\end{equation*}
	is trivial in view of the definition of the variational sets.
	
	To prove the opposite inequality, let us fix a minimizer $\mu_\varepsilon$ for 				\eqref{eq:energy_perturbed_constraint}. We will construct a suitable trial function for 		the problem with $\varepsilon=0$ by starting from $\mu_\varepsilon$ and removing all the 	mass that exceeds the upper constraint $B/2\pi\ell$. Define
	\begin{equation*}
	\Omega_\varepsilon=\Big\{\mu_\varepsilon>\frac{B}{2\pi\ell}\Big\}
	\end{equation*}
	Notice that
	\begin{equation*}
	1=\int_{\mathbb{R}^2}\mu_\varepsilon\ge \int_{\Omega_\varepsilon}\mu_\varepsilon>\frac{B}{2\pi\ell}|\Omega_\varepsilon|
	\end{equation*}
	and therefore, by defining
	\begin{equation*}
	f_\varepsilon=\mu_\varepsilon-\frac{B}{2\pi\ell},
	\end{equation*}
	we have
	\begin{equation*}
	\int_{\Omega_\varepsilon}f_\varepsilon\le|\Omega_\varepsilon|\Big((1+\varepsilon)\frac{B}{2\pi\ell}-\frac{B}{2\pi\ell}\Big)\le \varepsilon.
	\end{equation*}
	Let $K$ be a measurable compact set such that 
	$$ \mu_\eps \leq \frac{B}{2\pi \ell} - \frac{1}{|K|} \int_{\Omega_{\eps}} f_\eps$$
	on $K$. Clearly such a set must exist, in view of the normalization of $\mu_\eps$. Define
	\begin{equation*}
	\widetilde\mu=\mu_\varepsilon-\mathbbm{1}_{\Omega_\varepsilon}f_\varepsilon+\frac{\mathbbm{1}_K}{|K|}\int_{\Omega_\varepsilon}f_\varepsilon.
	\end{equation*}
	By construction it satisfies
	\begin{equation*}
	\int\widetilde\mu=1-\int_{\Omega_\varepsilon}f_\varepsilon+\int_{\Omega_\varepsilon}f_\varepsilon=1
	\end{equation*}
	and
	\begin{equation*}
	\widetilde\mu=\begin{cases}
	\mu_\varepsilon+\frac{\mathbbm{1}_K}{|K|}\int_{\Omega_\varepsilon}f_\varepsilon\le \frac{B}{2\pi\ell}\qquad&\text{on }\mathbb{R}^2\setminus \Omega_\varepsilon\\
	\frac{B}{2\pi\ell}\qquad&\text{on } \Omega_\varepsilon
	\end{cases}
	\end{equation*}
	It is thus a suitable trial function for the minimization problem with $\varepsilon=0$.
Hence
\begin{align*}
E^\mathrm{flo}&\le\mathcal{E}^\mathrm{MF}[\widetilde\mu]=\mathcal{E}^\mathrm{MF}[\mu_\varepsilon]-\int_{\Omega_\varepsilon}vf_\varepsilon+\Big(\int_{\Omega_\varepsilon}f_\varepsilon\Big)\frac{1}{|K|}\int_K v\\
&+\lambda\iint_{\mathbb{R}^2\times\mathbb{R}^2}w(x-y)\Big[-\mu_\varepsilon(x)f_\varepsilon(y)\mathbbm{1}_{\Omega_\varepsilon}(y) +\Big(\int_{\Omega_\varepsilon}f_\varepsilon\Big)\mu_\varepsilon(x)\frac{\mathbbm{1}_K(y)}{|K|}\\
&-\Big(\int_{\Omega_\varepsilon}f_\varepsilon\Big)f_\varepsilon(x)\mathbbm{1}_{\Omega_\varepsilon}(x)\frac{\mathbbm{1}_K(y)}{|K|} +2f_\varepsilon(x)\mathbbm{1}_{\Omega_\varepsilon}(x)f_\varepsilon(y)\mathbbm{1}_{\Omega_\varepsilon}(y)+2\Big(\int_{\Omega_\varepsilon}f_\varepsilon\Big)^2\frac{\mathbbm{1}_K(x)\mathbbm{1}_K(y)}{|K|^2}\Big]dxdy.
\end{align*}
It is easy to see, using \eqref{eq:bounded_convolution} for the $w$-terms, that every summand except for $\mathcal{E}^\mathrm{MF}[\mu_\varepsilon]$ can be estimated from above by
\begin{equation*}
C\int_{\Omega_\varepsilon}f_\varepsilon\le C\varepsilon.
\end{equation*}
This implies
\begin{equation*}
E^\mathrm{flo}\le\mathcal{E}^\mathrm{MF}[\mu_\varepsilon]+C\varepsilon=E^\mathrm{flo}_\varepsilon+C\varepsilon,
\end{equation*}
which concludes the proof.
\end{proof}

For the minimization problem with perturbed upper constraint, we will also need to consider a functional with an external potential which is modified at large distances. For any $L>0$ let us denote by $\mathcal{S}_L$ the square of side $2L$ centered at the origin, 
\begin{equation*}
\mathcal{S}_L=[-L,L]^{2}
\end{equation*}
By our assumption \eqref{eq:growth_trap}, if $L$ is large enough, we can certainly construct a new potential $U_L$ of the type
\begin{equation*}
U_L(x)=\begin{cases}
v(x)\quad&x\in\mathcal{S}_L\\
|x|\quad&x\notin\mathcal{S}_{L+1}\\
C^1(\mathbb{R}^2)\text{-interpolation between }v(x)\text{ and }|x|\quad&x\in\mathcal{S}_{L+1}\setminus\mathcal{S}_L
\end{cases}
\end{equation*}
with the properties
\begin{align}
v\ge&\; U_L \label{eq:v_vs_U_L}\\ 
\label{eq:gradient_U_L} \|\nabla U_L\|_\infty \le&\; C_L\\
\label{eq:growth_U_L} U_L(x)\ge |x|\ge&\; L,\qquad\text{for }x\notin \mathcal{S}_L.
\end{align}

Let us define, for any probability measure $\mu$ such that the integrals make sense,
\begin{equation*}
\mathcal{E}^\mathrm{MF}_L[\mu]=\int U_Ld\mu+\frac{\lambda}{2}\int w(x-y)d\mu(x)d\mu(y).
\end{equation*}
We have the following result, proving that such a modification of the external potential does not play a role. Neither does  the removal of the upper constraint at large distances.

\begin{lemma}[\textbf{Perturbed problem at large distances}]\label{lemma:perturbed_L}\mbox{}\\
For $L$ large enough we have
\begin{equation} \label{eq:large_distance_perturbed_problem}
\inf\Big\{\mathcal{E}_L^\mathrm{MF}[\mu]\;|\;\mu\in\mathcal{P}(\mathbb{R}^2),\; (1+\varepsilon)\frac{B}{2\pi\ell}dx-\mu\ge0\text{ on }\mathcal{S}_L\Big\}=E^{\mathrm{flo}}_\varepsilon.
\end{equation}
\end{lemma}

Notice that we extended the minimization set to \emph{generic} probability measures which are bounded above by the Lebesgue measure when restricted to $\mathcal{S}_L$ in order to ensure the existence of minimizers. However, this means that trial measures for the problem \eqref{eq:large_distance_perturbed_problem} do not have a density in general (because they might have an atom outside of $\mathcal{S}_L$). They nevertheless do have a density when restricted to $\mathcal{S}_L$, and such a density is bounded above by $B/2\pi\ell$. For any $\mu$, we will call that density $\mu_{|\mathcal{S}_L}(x)$.

\begin{proof}
We have
\begin{equation*}
\begin{split}
\inf\Big\{\mathcal{E}_L^\mathrm{MF}[\mu]&\,|\,\mu\in\mathcal{P}(\mathbb{R}^2),\; (1+\varepsilon)\frac{B}{2\pi\ell}dx-\mu\ge0\text{ on }\mathcal{S}_L\Big\}\\
\le\;&\inf\Big\{\mathcal{E}^\mathrm{MF}[\mu]\,|\,\mu\in\mathcal{P}(\mathbb{R}^2),\; (1+\varepsilon)\frac{B}{2\pi\ell}dx-\mu\ge0\text{ on }\mathcal{S}_L\Big\}\\
\le\;&\inf\Big\{\mathcal{E}^\mathrm{MF}[\mu]\,|\,\mu\in\mathcal{P}(\mathbb{R}^2),\, (1+\varepsilon)\frac{B}{2\pi\ell}dx-\mu\ge0\Big\}\;=\;E^\mathrm{flo}_\varepsilon,
\end{split}
\end{equation*}
where the first step follows from $U_L\le v$ and the second one is due to the fact that adding the constraint at large distances reduces the variational set.

Let us prove the opposite inequality. Let the measure $\mu_m^{(L)}$ be a minimizer for the $L$-dependent problem, i.e.,
\begin{equation} \label{eq:minimization_L}
\inf\Big\{\mathcal{E}_L^\mathrm{MF}[\mu]\,|\,\mu\in\mathcal{P}(\mathbb{R}^2),\; (1+\varepsilon)\frac{B}{2\pi\ell}dx-\mu\ge0\text{ on }\mathcal{S}_L\Big\}=\mathcal{E}^\mathrm{MF}_L[\mu_m^{(L)}].
\end{equation}
The existence of such a minimizer follows by the same arguments as previously (notice that by construction $U_L$ is lower semicontinuous). Assume that, for $L$ large enough, the support of $\mu^{(L)}_m$ is entirely contained in $\mathcal{S}_L$. Then we deduce
\begin{equation*}
\mathcal{E}^\mathrm{MF}_L[\mu_m^{(L)}]=\mathcal{E}^\mathrm{MF}[\mu_m^{(L)}]\ge E^\mathrm{flo}_\varepsilon,
\end{equation*}
where the first equality is due to $v=U_L$ on $\mathcal{S}_L$ and the second one is the variational principle. This completes the proof in view of~\eqref{eq:minimization_L}. 

We thus aim at proving that, for $L$ large enough, the support of $\mu^{(L)}_m$ is entirely contained in $\mathcal{S}_L$. 
Assume for contradiction that, for a sequence of $L$'s accumulating at $+\infty$, there exists a (sequence of) set(s) $\Sigma_L$ of positive measure, entirely contained in $\mathbb{R}^2\setminus \mathcal{S}_L$, and such that 
$$\int_{\Sigma_L}d\mu_m^{(L)}>0.$$
We will show that moving mass from $\Sigma_L$ to the interior of $\mathcal{S}_L$ decreases the energy, hence violating minimality of $\mu_m^{(L)}$. To do so, we first show that $\mu_m^{(L)}$ does not saturate the upper constraint on the whole $\mathcal{S}_L$. Indeed, if there exists $c\in(0,(1+\varepsilon)B/2\pi\ell]$ such that the (Lebesgue) measure of the set $\{x\in\mathcal{S}_L\,|\,\mu_{m\,|\mathcal{S}_L}^{(L)}(x)\ge c\}$ tends to infinity as $L\to\infty$, then integrability of $\mu_m^{(L)}$ is impossible. 

Hence, for every $c\in(0,(1+\varepsilon)B/2\pi\ell]$,
\begin{equation} \label{eq:L_uniformity}
\big|\{x\in\mathcal{S}_L\,|\,\mu_{m\,|\mathcal{S}_L}^{(L)}(x)\ge c\}\big|\le C \quad\text{uniformly in }L.
\end{equation}
Let us fix $c\in(0,(1+\varepsilon)B/2\pi\ell]$. As a consequence of \eqref{eq:L_uniformity}, for every $L$ large enough there exists a set $R_L$ with the following properties:
\begin{enumerate}
	\item $\mu_m^{(L)}(x)<c$ on $R_L$;
	\item $|R_L|=C$ independent of $L$;
	\item each $R_L$ is entirely contained in a ball whose radius is a constant independent of $L$, i.e., there exists $C>0$ independent of $L$ such that
	\begin{equation} \label{eq:R_L_bounded}
	R_L\subset \big\{x\in\mathbb{R}^2\,|\,|x|\le C\big\}\quad\forall L.
	\end{equation}
\end{enumerate}
The catch is that, even though the set on which $\mu_m^{(L)}$ approximates its upper bound might `pulsate' when the value of $L$ changes, by occupying different regions of $\mathcal{S}_L$, still, there will always be a set where $\mu_m^{(L)}$ is far from the upper constraint and this set can always be picked within the same fixed bounded region.

Define the measure
\begin{equation*}
\widetilde\mu=\mu_m^{(L)}-\mu_{m\,|\Sigma_L}^{(L)}+\frac{\mathbbm{1}_{R_{L}}dx}{|R_{L}|}\int_{\Sigma_L}d\mu_m^{(L)}.
\end{equation*}
By definition we have
\begin{equation*}
\int_{\mathbb{R}^2} d\widetilde\mu=1-\int_{\Sigma_L}d\mu_m^{(L)}+\int_{\Sigma_L}d\mu_m^{(L)}=1.
\end{equation*}
Moreover,
\begin{equation*}
\widetilde\mu=\begin{cases}
\mu_m^{(L)}\qquad &\text{on } \mathbb{R}^2\setminus(\Sigma_L\cup R_{L})\\
0\qquad&\text{on } \Sigma_L\\
\mu_m^{(L)}+\frac{\mathbbm{1}_{R_{L}}dx}{|R_{L}|}\int_{\Sigma_L}d\mu_m^{(L)}<cdx+\frac{\mathbbm{1}_{R_{L}}dx}{|R_{L}|}\int_{\Sigma_L}d\mu_m^{(L)}\quad&\text{on }R_{L}.
\end{cases}
\end{equation*}
We may now choose $\Sigma_L$ small enough so that $\widetilde\mu_{|\mathcal{S}_L}(x)\le (1+\varepsilon)B/2\pi\ell$, which makes it an admissible trial measure for the minimization problem \eqref{eq:minimization_L}. Therefore
\begin{equation} \label{eq:minimality_L}
\mathcal{E}_L^\mathrm{MF}[\widetilde\mu]\ge \mathcal{E}_L^\mathrm{MF}[\mu_m^{(L)}].
\end{equation}
However we also have
\begin{equation*}
\begin{split}
\mathcal{E}_L^\mathrm{MF}[\widetilde\mu]=\;&\mathcal{E}_L^\mathrm{MF}[\mu^{(L)}_m]-\int_{\Sigma_L} U_Ld\mu_m^{(L)}+\frac{\int_{\Sigma_L}d\mu_m^{(L)}}{|R_{L}|}\int_{R_{L}} U_L(x)dx\\
&+\lambda\iint_{\mathbb{R}^2\times\mathbb{R}^2}w(x-y)\Big[-d\mu_m^{(L)}(x)d\mu_{m\,|\Sigma_L}^{(L)}(y)+\Big(\int_{\Sigma_L}d\mu_m^{(L)}\Big)d\mu_m^{(L)}(x)\frac{\mathbbm{1}_{R_{L}}(y)dy}{|R_{L}|}\\
&\qquad\qquad\qquad\qquad\quad\quad-\Big(\int_{\Sigma_L}d\mu_m^{(L)}\Big)d\mu_{m\,|\Sigma_L}^{(L)}(x)\frac{\mathbbm{1}_{R_{L}}(y)dy}{|R_{L}|}\\
&\qquad\qquad\qquad\qquad\quad\quad+2d\mu_{m\,|\Sigma_L}^{(L)}(x)d\mu_{m\,|\Sigma_L}^{(L)}(y)\\
&\qquad\qquad\qquad\qquad\quad\quad+2\Big(\int_{\Sigma_L}d\mu_m^{(L)}\Big)^2\frac{\mathbbm{1}_{R_{L}}(x)dx}{|R_{L}|}\frac{\mathbbm{1}_{R_{L}}(y)dy}{|R_{L}|}\Big].
\end{split}
\end{equation*}
Since $w$ is bounded, it is easy to see that every term containing $w$ can be bounded from above by
\begin{equation*}
C\int_{\Sigma_L}d\mu_m^{(L)}.
\end{equation*}
Moreover, by \eqref{eq:growth_U_L} we know that $U_L(x)\ge L$ for $x\in\Sigma_L$, and $U_L(x)=v(x)\le C$ for $x\in R_{L}$ by \eqref{eq:R_L_bounded}. We then deduce
\begin{equation*}
\mathcal{E}_L^\mathrm{MF}[\widetilde\mu]\le \mathcal{E}_L^\mathrm{MF}[\mu^{(L)}_m]-L\int_{\Sigma_L}d\mu_m^{(L)}+C\int_{\Sigma_L}d\mu_m^{(L)},
\end{equation*}
which, for $L$ large enough, implies
\begin{equation*}
\mathcal{E}_L^\mathrm{MF}[\widetilde\mu]< \mathcal{E}_L^\mathrm{MF}[\mu^{(L)}_m].
\end{equation*}
This contradicts \eqref{eq:minimality_L}, and the proof is complete. 
\end{proof}

\section{Energy lower bound} \label{sect:lower_bound}

The aim of this section is the proof of Theorem~\ref{thm:energies low}. Given $\Psi_F\in\mathcal{L}_{\ell,B}^N$, let $\mu_F$ be the associated $N$-body probability measure in rescaled units of lengths, i.e.,
\begin{equation} \label{eq:mu_F}
\mu_F(x_1,\dots,x_N)=N^N\big|\Psi_F(\sqrt{N}x_1,\dots\sqrt{N}x_N)\big|^2.
\end{equation}
Define the $\mu_F$-probability of a Borel subset (event) $\Gamma\subset\mathbb{R}^{2N}$ as
\begin{equation*} \label{eq:prob_F}
\mathbb{P}_F(\Gamma)=\int_\Gamma\mu_F.
\end{equation*}
For any configuration $X_N=(x_1,\dots,x_N)\in\mathbb{R}^{2N}$ we will denote the corresponding empirical measure by
\begin{equation*}
\Emp_{X_N}=\frac{1}{N}\sum_{i=1}^N\delta_{x_i}.
\end{equation*}

The following statement gives an estimate on the probability of configurations whose empirical measure violates the incompressibility bound \eqref{eq:incompressibility}.

\begin{theorem}[\textbf{Probability of violating the incompressibility bound}] \label{thm:exponential_bound}\mbox{}\\
Let $\Omega$ denote any open set with Lipschitz boundary, and, for any $\alpha>(-3+\sqrt{5})/4$, let $\Omega_r$ denote its dilation around some origin by a factor $r=N^\alpha$. Let $\mu_F$ be the probability measure associated to $\Psi_F\in\mathcal{L}^N_{\ell,B}$. Then, for any $\varepsilon>0$,
\begin{equation} \label{eq:exponential_bound}
\mathbb{P}_F\left(\left\{X_N\in\mathbb{R}^{2N}\,\big|\, \int_{\Omega_r}\Emp_{X_N}>(1+\varepsilon)\frac{B|\Omega_r|}{2\pi\ell} \right\}\right)\le e^{-C\varepsilon N^{\sqrt{5}-1}}.
\end{equation}
\end{theorem}

The constant $C$ in~\eqref{eq:exponential_bound} depends on the geometrical details of the unscaled set $\Omega$. A key ingredient for the proof is the incompressibility bound \cite[Theorem 2.1]{LieRouYng-17}, which in the notation of Theorem \ref{thm:exponential_bound} reads
\begin{equation} \label{eq:incompressibility_mu}
\int_{\Omega_r} \mu_F^{(1)}\le\frac{B|\Omega_r|}{2\pi\ell}(1+o_N(1)),
\end{equation}
where
\begin{equation*}
\mu_F^{(1)}(x)=\int\mu_F(x,x_2,\dots,x_N)dx_2\dots dx_N.
\end{equation*}
Note that~\eqref{eq:incompressibility_mu} says that
$$ \mathbb{E} \left( \int_{\Omega_r}  \Emp_{X_N} \right) \leq \frac{B|\Omega_r|}{2\pi\ell}(1+o_N(1))$$
and that our new estimate~\eqref{eq:exponential_bound} is a deviation bound for this expectation bound.

\begin{remark}[Scales on which incompressibility holds]\mbox{}\\
During the preparation of this work we became aware of an imprecision in the paper \cite{LieRouYng-17}. It is claimed in Theorem 2.1 there, that the incompressibility bound, in the rescaled version \eqref{eq:incompressibility_mu}, holds for scales up to $\alpha>1/4$. This does not however follow from the proof of Theorem 2.1 which is presented in \cite{LieRouYng-17}.

More precisely, after (5.50) there, one should actually deduce (5.51) in the form
\begin{equation*}
\gamma=\frac{1}{2+2\alpha'}+\frac{(2+\delta)\beta}{2+2\alpha'}.
\end{equation*}
After choosing $\alpha'$ and $\beta$ arbitrarily close to $\alpha$, and $\delta$ arbitrarily small, one finds that $\gamma$ can be picked arbitrarily close to $(1+2\alpha)/(2+2\alpha)$. Hence, the remainder terms in the right hand side of (5.50) are of order $O(N^{-1/(2+2\alpha)+\varepsilon})$. By comparing this with the main term, which is of order $N^{-1+2\alpha}$, one deduces that the remainders are actually negligible only for $\alpha>(-1+\sqrt{5})/4$. By rescaling units of lengths one deduces that \eqref{eq:incompressibility_mu} holds precisely for the range of $\alpha$ which is considered in Theorem~\ref{thm:exponential_bound}, i.e., $\alpha>(-3+\sqrt{5})/4$.

This small imprecision is only due to the localization procedure in~\cite[Section~5.2]{LieRouYng-17}. The bounds of~\cite[Section~5.1]{LieRouYng-17}, valid under small additional localization assumptions, are correct as stated.   \hfill$\diamond$
\end{remark}

\begin{proof}
Let us consider the case in which $|\Omega_r|$ is a disk of radius $r$, and let $\chi_r$ be the corresponding characteristic function.  To deduce the result for more general sets one simply replaces characteristic functions of disks with those of general open sets. We denote by $\chi_{r,\delta}$ the regularization of $\chi_r$ on a scale $\delta$, chosen so that
\begin{equation} \label{eq:regularized_laplacian}
\|\Delta\chi_{r,\delta}\|_\infty\le C \delta^{-2}
\end{equation}
and
\begin{equation} \label{eq:regularization}
\chi_r\le \chi_{r,\delta}\le \chi_{r+\delta}.
\end{equation}

We first rewrite $\mu_F$ as a Gibbs measure. By the definition \eqref{eq:mu_F} we notice that
\begin{equation} \label{eq:correspondence}
\mu_F(x_1,\dots,x_N)=\frac{1}{\mathcal{Z}}\exp(-N\mathbb{H}(x_1,\dots,x_N)),
\end{equation}
with
\begin{equation*}
\mathbb{H}(x_1,\dots,x_N)=\sum_{j=1}^N|x_j|^2-\frac{4\ell}{BN}\sum_{1\le i<j\le N}\log|x_i-x_j|-\frac{4}{BN}\log F\left(\sqrt{N}x_1,\dots,\sqrt{N}x_N \right)
\end{equation*}
and the normalization factor
\begin{equation*}
\mathcal{Z}=\int_{\R ^{dN}} \exp\left(-N\mathbb{H}(x_1,\dots,x_N)\right)dx_1\dots dx_N
\end{equation*}
is the corresponding partition function. Hence $\mu_F$ is the Gibbs measure for $\mathbb{H}$ with temperature $N^{-1}$.

Let us define the perturbed Hamiltonian
\begin{equation*}
\begin{split}
\mathbb{H}_s(x_1,\dots,x_N)=\mathbb{H}(x_1,\dots,x_N)-\frac{s}{N^2}\sum_{i=1}^N\chi_{r,\delta}(x_i),
\end{split}
\end{equation*}
the associated partition function at temperature $N^{-1}$
\begin{equation*}
\mathcal{Z}_s=\int \exp(-N\mathbb{H}_s(x_1,\dots,x_N))dx_1\dots dx_N.
\end{equation*}
and the associated Gibbs measure
$$ 
\mu_s := \frac{1}{\cZ_s} \exp(-N\mathbb{H}_s(x_1,\dots,x_N)).
$$
According to this notation $\mu_{0}=\mu_F$.

We also introduce the free-energy functional
\begin{equation*}
\mathcal{F}_s[\mu]=\int \mathbb{H}_s(X_N)\mu(X_N)+\frac{1}{N}\int\mu(X_N)\log\mu(X_N)
\end{equation*}
with
\begin{equation*}
F_s=\inf\left\{\mathcal{F}_s[\mu]\,|\,\int\mu=1,\,\mu\ge0\right\}.
\end{equation*}
Note that, for any positive function $\mu$ on $\R^{2N}$ with unit integral 
$$ \mathcal{F}_s[\mu] = \mathcal{F}_s[\mu_s] + N^{-1} \int_{\R^{2N}} \mu \left( \log\mu -\log \mu_s \right)$$
and Jensen's inequality implies that the last term is $\geq 0$ with equality if and only if $\mu = \mu_s$. This proves the classical fact that $\mu_s$ minimizes $\cF_s$ and it follows that 
$$F_s=-N^{-1}\log\mathcal{Z}_s,$$
which is the classical relation between free energy, temperature, and partition function. 

Let us show that $\mu_s^{(1)}$ satisfies an incompressibility bound analogous to \eqref{eq:incompressibility_mu}. By adding and subtracting to $\mathbb{H}_s$ the term
\begin{equation*}
\frac{s}{4N^2}\|\Delta\chi_{r,\delta}\|_\infty\sum_{j=1}^N|x_j|^2
\end{equation*}
we write
\begin{equation*}
\mathbb{H}_s(x_1,\dots,x_N)=\Big(1+\frac{s}{4N^2}\|\Delta\chi_{r,\delta}\|_\infty\Big)\Big[\sum_{j=1}^N|x_j|^2-\frac{4\widetilde\ell}{BN}\sum_{i<j}\log|x_i-x_j|\Big]+W(x_1,\dots,x_N),
\end{equation*}
where $\widetilde\ell=(1+\frac{s}{4N^2}\|\Delta\chi_{r,\delta}\|_\infty)^{-1}\ell$ and 
$$W (x_1,\ldots,x_N)  = -\frac{4}{BN}\log F\left(\sqrt{N}x_1,\dots,\sqrt{N}x_N \right) -\frac{s}{N^2}\sum_{j=1}^N\Big(\chi_{r,\delta}(x_j)+\frac{1}{4}\|\Delta\chi_{r,\delta}\|_\infty|x_j|^2\Big)
$$ 
is superharmonic in each variable because the last two terms taken together are. We recognize that, apart from an irrelevant multiplicative factor, $\mathbb{H}_s$ has the same structure that $\mathbb{H}$ has, with a modified $\ell$ and a different superharmonic function. Hence, one can replicate the proof of \eqref{eq:incompressibility_mu} from \cite{LieRouYng-17} so as to get
\begin{equation} \label{eq:modified_incompressibility}
\int_{\Omega_r}\mu^{(1)}_s\le\frac{B}{2\pi\ell}|\Omega_r|\big(1+C\frac{s}{\delta^2 N^2}+o_N(1)\big).
\end{equation}
Now, we notice that, for any $s>0$,
\begin{equation} \label{eq:prob_less_exp}
\begin{split}
\mathbb{P}_F\Big(\Big\{X_N\in\mathbb{R}^{2N}&\,|\, \int_{\Omega_r}\Emp_{X_N}>(1+\varepsilon)\frac{B|\Omega_r|}{2\pi\ell}\Big\}\Big)\\
\le\;& e^{-s\frac{B|\Omega_r|}{2\pi\ell}(1+\varepsilon)}\int \exp\Big(s\int_{\Omega_r}\Emp_{X_N}\Big)d\mu_F(X_N)\\
\le\;& e^{-s\frac{B|\Omega_r|}{2\pi\ell}(1+\varepsilon)}\int \exp\Big(s\int\chi_{r,\delta\,}\Emp_{X_N}\Big)d\mu_F(X_N),
\end{split}
\end{equation}
where the first inequality follows from the bound $1\le e^{sx}$, which holds whenever $sx>0$, and the second inequality comes from \eqref{eq:regularization}. Notice that \eqref{eq:prob_less_exp} is Markov's inequality applied to $\exp\{s \int\chi_{r,\delta}\mathrm{Emp}_{X_N}\}$. We will focus on estimating the integral on the right.

Within the above notation we recognize that
\begin{equation} \label{eq:partition_functions}
\int \exp\Big(s\int\chi_{r,\delta\,}\Emp_{X_N}\Big)d\mu_F(X_N)=\frac{\mathcal{Z}_s}{\mathcal{Z}_0}=e^{-N(F_s+F_0)}.
\end{equation}
We estimate $F_s$ by writing
\begin{equation*}
\begin{split}
F_s=\;&\int\mathbb{H}\mu_s-\frac{s}{N}\int\chi_{r,\delta}\mu_s^{(1)}+\frac{1}{N}\int\mu_s\log\mu_s\\
\ge\;& F_0-\frac{s}{N}\int\chi_{r,\delta}\mu_s^{(1)}\\
\ge \;& F_0-\frac{s}{N}\int\chi_{r+\delta}\mu_s^{(1)}\\
\geq\;&F_0-\frac{s}{N}\int_{\Omega_{r+\delta}}\mu_s^{(1)},
\end{split}
\end{equation*}
where the first inequality holds by minimality of $F_0$ and the second one follows from \eqref{eq:regularization}. Hence \eqref{eq:partition_functions} yields, using \eqref{eq:modified_incompressibility},
\begin{equation} \label{eq:bound_expectation}
\begin{split}
\int \exp\Big(s\int\chi_{r,\delta\,}\Emp_{X_N}\Big)d\mu_F(X_N)\le\;& \exp\Big(s\int_{\Omega_{r+\delta}} \mu^{(1)}_s \Big)\\
\le\;&\exp\Big(s\frac{B(r+\delta)^2}{2\ell}\big(1+C\frac{s}{\delta^2N^2}+o_N(1)\big)\Big)\\
\le\;&\exp\Big(s\frac{B|\Omega_r|}{2\pi\ell}\big(1+o_N(1)\big)\Big),
\end{split}
\end{equation}
provided we choose
\begin{equation*}
r\gg\delta \quad\text{and}\quad 1\gg \frac{s}{\delta^2N^2}.
\end{equation*}
Since $r=N^{\alpha}$, the above relations are satisfied for $s=N^\beta$ with $\beta<2\alpha+2$.
Plugging \eqref{eq:bound_expectation} into \eqref{eq:prob_less_exp} yields
\begin{equation*}
\mathbb{P}_F\Big(\Big\{X_N\in\mathbb{R}^{2N}\,|\, \int_{\Omega_r}\Emp_{X_N}>(1+\varepsilon)\frac{B|\Omega_r|}{2\pi\ell}\Big\}\Big)\le \exp\Big(-\varepsilon N^\beta\frac{B|\Omega_r|}{2\pi\ell}(1+o_N(1)\Big).
\end{equation*}
By recalling that $|\Omega_r|=\pi N^{2\alpha}$ and $\alpha>(-3+\sqrt{5})/4$, we deduce that we can always choose $\beta$, depending on $\alpha$, so that
\begin{equation*}
\mathbb{P}_F\Big(\Big\{X_N\in\mathbb{R}^{2N}\,|\, \int_{\Omega_r}\Emp_{X_N}>(1+\varepsilon)\frac{B|\Omega_r|}{2\pi\ell}\Big\}\Big)\le e^{-C\varepsilon N^{\sqrt{5}-1}},
\end{equation*}
which proves the result.
\end{proof}

We now provide the

\begin{proof}[Proof of Theorem~\ref{thm:energies low}]
Let us divide $\mathbb{R}^2$ into squares of side $N^\alpha$ with $0> \alpha>(-3+\sqrt{5})/4$, and denote the $k$-th square by $S_k$. Let us consider a number $L>0$ of the form 
$$L=n N^\alpha$$
for some $n\in\mathbb{N}$, and define by $\mathcal{S}_L$ the square $[-L,L]^{2}$ of side $2L$ and centered at the origin. Notice that the number of $S_k$-squares within the set $\mathcal{S}_L$ is $4n^2$. We will later choose $L$ large enough but independent on $N$, which boils down to suitably choosing the (large) integer $n$. Let us also define
\begin{equation*}
A_L=\Big\{X_N\in \mathbb{R}^{2N}\,|\,\int_{S_k}\Emp_{X_N}\le (1+\varepsilon)\frac{B|S_k|}{2\pi\ell}\quad\forall k\,\text{ such that }\,S_k\subset\mathcal{S}_L\Big\},
\end{equation*}
and $A_L^c=\mathbb{R}^2\setminus A_L$. 

We start by estimating the measure of the set $A^c_L$. We have
\begin{equation*}
\mu_F(A^c_L)=\mathbb{P}_{F}\left(\bigcup_{S_k\subset\mathcal{S}_L} \left\{ \int_{S_k}\Emp_{X_N}>(1+\varepsilon)\frac{B|S_k|}{2\pi\ell}\right\} \right),
\end{equation*}
i.e., the measure of $A^c_L$ is the probability of the union of all events of the type ``the integral over the $k$-th square is bigger than $(1+\varepsilon)B|S_k|/2\pi\ell$''.
Hence, by the union bound and by Theorem~\ref{thm:exponential_bound},
\begin{equation} \label{eq:union_bound}
\begin{split}
\mu_F(A^c_L)\le\;&\sum_{k=1}^{4n^2}\mathbb{P}_{F}\Big(\int_{S_k}\Emp_{X_N}>(1+\varepsilon)\frac{B|S_k|}{2\pi\ell}\Big)\\
\le\;&\sum_{k=1}^{4n^2} \exp\big(-C\varepsilon N^{\sqrt{5}-1}\big)\\
=\;&CL^2N^{-2\alpha}\exp\big(-C\varepsilon N^{\sqrt{5}-1}\big).
\end{split}
\end{equation}

Now, by a straightforward computation we can express the many-body energy in the form
\begin{equation*}
\begin{split}
N^{-1}\mathcal{E}_{N,\lambda}[\Psi_F]=&\int\mathcal{E}^{\mathrm{MF}}[\Emp_{X_N}]d\mu_F(X_N)-\frac{\lambda}{2N}\iint_{\mathbb{R}^2\times\mathbb{R}^2} w(x-y)\mu_F^{(2)}(x,y)dxdy-\lambda\frac{w(0)}{2N}\\
\ge&\int\mathcal{E}^{\mathrm{MF}}[\Emp_{X_N}]d\mu_F(X_N)-CN^{-1},
\end{split}
\end{equation*}
having used $\|w\|_\infty\le C$ and $\int\mu^{(2)}_F=1$ in the second step. We split the integral by writing
\begin{equation*}
\begin{split}
\int\mathcal{E}^{\mathrm{MF}}[\Emp_{X_N}]d\mu_F(X_N)=\;& \int_{A_L}\mathcal{E}^{\mathrm{MF}}[\Emp_{X_N}]d\mu_F(X_N)+\int_{A_L^c}\mathcal{E}^{\mathrm{MF}}[\Emp_{X_N}]d\mu_F(X_N)\\
\ge\;& \int_{A_L}\mathcal{E}^{\mathrm{MF}}[\Emp_{X_N}]d\mu_F(X_N) - C \mu_F (A_L^c),
\end{split}
\end{equation*}
where we used that the potentials in $\mathcal{E}^\mathrm{MF}$ are uniformly bounded below in the second step. As a consequence we find
\begin{equation} \label{eq:first_splitting}
\begin{split}
N^{-1}\mathcal{E}_{N,\lambda}[\Psi_F]\ge\int_{A_L}\mathcal{E}^{\mathrm{MF}}[\Emp_{X_N}]d\mu_F(X_N)-CN^{-1} - C \mu_F (A_L^c).
\end{split}
\end{equation}
By \eqref{eq:growth_trap}, for $L$ large enough, we can consider $U_L$ of the type
\begin{equation*}
U_L(x)=\begin{cases}
v(x)\quad&x\in\mathcal{S}_L\\
|x|\quad&x\notin\mathcal{S}_{L+1}\\
C^1(\mathbb{R}^2)\text{-interpolation between }v(x)\text{ and }|x|\quad&x\in\mathcal{S}_{L+1}\setminus\mathcal{S}_L
\end{cases}
\end{equation*}
satisfying the properties \eqref{eq:v_vs_U_L}, \eqref{eq:gradient_U_L}, \eqref{eq:growth_U_L}.
As in Section \ref{sect:flocking}, we denote by $\mathcal{E}^\mathrm{MF}_L$ the functional $\mathcal{E}^\mathrm{MF}$ with $v$ replaced by $U_L$. Since $v\ge U_L$ we can write
\begin{equation*}
\begin{split}
N^{-1}\mathcal{E}_{N,\lambda}[\Psi_F]\ge\int_{A_L}\mathcal{E}^{\mathrm{MF}}_L[\Emp_{X_N}]d\mu_F(X_N)-CN^{-1} - C \mu_F (A_L^c).
\end{split}
\end{equation*}
We now perform a Riemann sum approximation on the whole of $\mathbb{R}^2$, by replacing the empirical measure $\Emp_{X_N}$ with the piece-wise constant function $\widetilde\mu^{(X_N)}$ whose value in the $k$-th square is
\begin{equation*}
\widetilde\mu^{(X_N)}_{| S_k}:=\frac{\int_{S_k}\Emp_{X_N}}{|S_k|}.
\end{equation*}
Notice that $\widetilde\mu^{(X_N)}$ satisfies
\begin{equation*}
\begin{split}
\int_{\mathbb{R}^2}\widetilde\mu^{(X_N)}&=1\\
 \widetilde\mu^{(X_N)}&\ge0\\
  \widetilde\mu^{(X_N)}&\le (1+\varepsilon)\frac{B}{2\pi\ell}\quad\text{on every }S_k\subset[-L,L]^2. 
\end{split}
\end{equation*}
This yields (recall that $L$ is independent of $N$)
\begin{equation} \label{eq:second_step}
\begin{split}
N^{-1}\mathcal{E}_{N,\lambda}[\Psi_F]\ge& \int_{A_L} \mathcal{E}^\mathrm{MF}_{L}[\widetilde\mu^{(X_N)}]d\mu_F(X_N)-CN^{\alpha}\big(\|\nabla w\|_{\infty}+\|\nabla U_L\|_\infty\big)-CN^{-1} - C \mu_F (A_L^c)\\
=& \int_{A_L} \mathcal{E}^\mathrm{MF}_{L}[\widetilde\mu^{(X_N)}]d\mu_F(X_N)-C_LN^{\alpha} - C \mu_F (A_L^c),
\end{split}
\end{equation}
where the constant $C_L$ in the right hand side is due to \eqref{eq:gradient_U_L}.

Due to the above properties of $\widetilde\mu^{(X_N)}$, we have
\begin{equation*}
\mathcal{E}^\mathrm{MF}_L[\widetilde\mu^{(X_N)}]\ge\inf \Big\{\mathcal{E}^\mathrm{MF}_L[\mu]\,|\,\mu\in\mathcal{P}(\mathbb{R}^2),\; (1+\varepsilon)\frac{B}{2\pi\ell}dx-\mu\ge0\text{ on }\mathcal{S}_L\Big\}=E^\mathrm{flo}_\varepsilon,
\end{equation*}
the equality having been proven in Lemma \ref{lemma:perturbed_L} for $L$ large enough but finite. Recall that $E^\mathrm{flo}_\varepsilon$ is the flocking energy with perturbed upper constraint defined in \eqref{eq:energy_perturbed_constraint}. Recall also that by Lemma \ref{lemma:dependence_upper_constraint}
\begin{equation*}
E^{\mathrm{flo}}_{\varepsilon}\ge E^\mathrm{flo}-C\varepsilon,
\end{equation*}
and therefore
\begin{equation*}
N^{-1}\mathcal{E}_{N,\lambda}[\Psi_F]\ge E^{\mathrm{flo}} \mu_F(A_L)-C_LN^{\alpha}-C\varepsilon - C \mu_F (A_L^c).
\end{equation*}
Using $\mu_F(A_L)=1-\mu_F(A_L^c)$ and then \eqref{eq:union_bound} we deduce
\begin{equation*}
N^{-1}\mathcal{E}_{N,\lambda}[\Psi_F]\ge E^{\mathrm{flo}} -CL^2N^{-2\alpha}\exp\big(-C\varepsilon N^{\sqrt{5}-1}\big) -C_LN^{\alpha}-C\varepsilon.
\end{equation*}
We choose $\varepsilon=N^{-1}$. In this way we can neglect the exponential because $L$ is large but fixed. Since $\alpha>(-3+\sqrt{5})/4$ we obtain
\begin{equation*}
N^{-1}\mathcal{E}_{N,\lambda}[\Psi_F]\ge E^{\mathrm{flo}}-CN^{(-3+\sqrt{5})/4+\gamma}
\end{equation*}
for any $\gamma>0$.
This concludes the proof.
\end{proof}

\section{Energy upper bound} \label{sect:upper_bound}

Now we prove Theorem~\ref{thm:energies up}. Let us write a generic polynomial $f$ as
\begin{equation*}
f(z)=c_N\prod_{j=1}^J(z-\sqrt{N}a_j)^{Nq_j/2}.
\end{equation*}
Given the corresponding $\Psi_f$ of the form \eqref{eq:restricted_class}, the associated $N$-particle probability density in rescaled units of length $\mu_f$ (see \eqref{eq:mu_F}) can be written as
\begin{equation}
\mu_f(x_1,\dots,x_N)=\frac{1}{\mathcal{Z}}e^{-N\mathcal{H}_N(x_1,\dots,x_N)},
\end{equation}
where the partition function $\mathcal{Z}$ is a normalization factor and
\begin{equation}
\mathcal{H}_N(x_1,\dots,x_N)=\sum_{i=1}^N\Big(\sum_{j=1}^Jq_j\log\frac{1}{|x_i-a_j|}+\frac{B}{2}|x_i|^2\Big)+\frac{2\ell}{N}\sum_{k<\ell}\log\frac{1}{|x_k-x_\ell|}.
\end{equation}
Hence $\mu_f$ is the Gibbs state for the Hamiltonian $\mathcal{H}_N$ at temperature $T=1/N$. $\mathcal{H}_N$ is the energy of $N$ two-dimensional particles interacting among themselves and with $J$ fixed point particles $a_j$ through the 2D-Coulomb kernel $-\log|\cdot|$, and trapped by an external harmonic potential. The charges $q_j$ will be eventually chosen to be equal to $2/N$, but the above notation will allow the reader to draw a direct comparison with the results of \cite[Section~3]{RouYng-17}.

The mean-field functional associated to the electrostatic Hamiltonian $\mathcal{H}_N$ is, for a given probability measure $\mu$,
\begin{equation} \label{eq:electrostatic_energy}
\mathcal{E}^\mathrm{el}_f[\mu]=\int_{\mathbb{R}^2}\left(\sum_{j=1}^Jq_j\log\frac{1}{|x-a_j|}+\frac{B}{2}|x|^2\right)\mu(x)dx+\ell\iint_{\mathbb{R}^2\times\mathbb{R}^2}\mu(x)\log\frac{1}{|x-y|}\mu(y)dxdy.
\end{equation}
As discussed in \cite[Sect. 3-4]{RouYng-17} and \cite{RouSerYng-13b}, as $N\to\infty$, there exists a unique normalized minimizer $\mu_f^{\mathrm{el}}$ satisfying
\begin{equation} \label{eq:electric_problem}
E^\mathrm{el}_f:=\inf\Big\{\mathcal{E}^\mathrm{el}[\mu]\;|\;\mu\ge0,\,\int\mu=1\Big\}=\mathcal{E}^\mathrm{el}_f[\mu_f^\mathrm{el}].
\end{equation}
Moreover, $\mu_f^{\mathrm{el}}$ only assumes the values zero and $B/2\pi\ell$. 

Let us define the 2D electrostatic energy
\begin{equation*}
D(\sigma,\sigma)=\frac{1}{2}\iint_{\mathbb{R}^2\times\mathbb{R}^2}\sigma(x)\log\frac{1}{|x-y|}\sigma(y)dxdy
\end{equation*}
for all $\sigma$ for which the integral is well defined (for example $\sigma$ such that $\int\left|\log|x|\right| \left|\sigma(x)\right|dx<+\infty$)

\begin{lemma}[\textbf{Bounds using the electrostatic energy}]\label{lemma:Coulomb_distance}\mbox{}\\
Let $\mu_1$ and $\mu_2$ be probability densities such that $D(\mu_1-\mu_2,\mu_1-\mu_2)<+\infty$. 
Then, for any test-function $\chi_1(x)$,
\begin{equation}  \label{eq:Coulomb_distance1}
\Big|\int_{\mathbb{R}^2}\big(\mu_1(x)-\mu_2(x))\chi_1(x)dx\Big|\le \|\nabla\chi_1\|_{L^2}D(\mu_1-\mu_2,\mu_1-\mu_2)^{1/2},
\end{equation}
and, for any test-function $\chi_2(x,y)$,
\begin{equation} \label{eq:Coulomb_distance2}
\begin{split}
\Big|\iint_{\mathbb{R}^2\times\mathbb{R}^2}\big(\mu_1(x)\mu_1(y)-\mu_2(x)\mu_2(y)\big)\chi_2(x,y)dxdy\Big|\le C&D(\mu_1-\mu_2,\mu_1-\mu_2)^{1/2}\\&\times\sup_y\|\nabla\chi_2(\cdot,y)\|_{L^2}
\end{split}
\end{equation}
\end{lemma}

\begin{proof}
We will prove the statement for smooth and compactly supported $\mu_1(x)$ and $\mu_2(x)$. The general result is then obtained by a density argument. Note that 
$$ D (\mu_1-\mu_2,\mu_1-\mu_2) \geq 0$$
when $\int_{\R^2} \mu_1 = \int_{\R^2} \mu_2$, see~\cite[Chapter I, Lemma~1.8]{SafTot-97}.
Assume for the moment that $\chi_1$ is differentiable and denote 
$$ V = -\frac{1}{2\pi}\log |\, . \,| \star \left( \mu_1 - \mu_2\right). $$ 
Then, since 
$$ -\Delta V = \mu_1 - \mu_2$$
\begin{equation*}
\int_{\mathbb{R}^2}\big(\mu_1(x)-\mu_2(x))\chi_1(x)dx=\int {\nabla\chi_1}(x)\cdot {\nabla V}(x)dx,
\end{equation*}
and by Cauchy-Schwarz
\begin{equation*}
\begin{split}
\Big|\int_{\mathbb{R}^2}\big(\mu_1(x)-\mu_2(x))\chi_1(x)dx\Big|\le\;& \|\nabla \chi_1\|_{L^2}\big\|\nabla V \big\|_{L^2}\\
=\;&C\|\nabla\chi_1\|_{L^2}\Big[\int_{{\mathbb{R}^2}}\Big(\nabla_x\int_{\mathbb{R}^2}\frac{1}{|x-y|}(\mu_1(y)-\mu_2(y))dy\Big)\\
&\qquad\qquad\qquad\qquad\times\Big(\nabla_x\int_{\mathbb{R}^2}\frac{1}{|x-z|}(\mu_1(z)-\mu_2(z))dz\Big)\Big]^{1/2}\\
=\;&C\|\nabla \chi_1\|_{L^2}\Big[\int_{{\mathbb{R}^2}}\int_{\mathbb{R}^2}\frac{1}{|x-y|}(\mu_1(y)-\mu_2(y))dy\\
&\qquad\qquad\qquad\qquad\times\Big(\Delta_x\int_{\mathbb{R}^2}\frac{1}{|x-z|}(\mu_1(z)-\mu_2(z))dz\Big)\Big]^{1/2}\\
=\;&\|\nabla\chi_1\|_{L^2} D(\mu_1-\mu_2,\mu_1-\mu_2)^{1/2}.
\end{split}
\end{equation*}
The result is then extended by density to test functions with $\nabla\chi_1\in L^2$.

To prove \eqref{eq:Coulomb_distance2} we can safely assume $\chi_2(x,y)=\chi_2(y,x)$, because the contribution of the antisymmetric part would integrate to zero. We have
\begin{equation*}
\begin{split}
\iint_{\mathbb{R}^2\times\mathbb{R}^2}&\big(\mu_{1}(x)\mu_1(y)-\mu_2(x)\mu_2(y)\big)\chi_2(x,y)dxdy\\
=&\iint_{\mathbb{R}^2\times\mathbb{R}^2}\big(\mu_1(x)+\mu_2(x)\big)\big(\mu_1(y)-\mu_2(y)\big)\chi_2(x,y)dxdy,
\end{split}
\end{equation*}
and hence, by \eqref{eq:Coulomb_distance1},
\begin{equation*}
\begin{split}
\Big|\iint_{\mathbb{R}^2\times\mathbb{R}^2}&\big(\mu_{1}(x)\mu_1(y)-\mu_2(x)\mu_2(y)\big)\chi_2(x,y)dxdy\Big|\\
&\le D(\mu_1-\mu_2,\mu_1-\mu_2)^{1/2}\;\Big\|\nabla\int\chi_2(\cdot,x)\big(\mu_1(x)+\mu_2(x)\big)dx\Big\|_{L^2}.
\end{split}
\end{equation*}
The Cauchy-Schwarz inequality together with $\int\mu_1=\int\mu_2=1$ yields
\begin{equation*}
\Big\|\nabla\int\chi_2(\cdot,x)\big(\mu_1(x)+\mu_2(x)\big)dx\Big\|_{L^2} \le C\sup_x\|\nabla\chi_2(\cdot,x)\|_{L^2},
\end{equation*}
which concludes the proof.
\end{proof}

The first main result of this section is the following proposition, showing that $\mu_f$ approximates $\mu^\mathrm{el}_f$ in the sense of reduced densities.

\begin{proposition}[\textbf{Mean-field approximation of $\Psi_{f}$}]\mbox{}\label{prop:mf_upper}\\ 
Let $f$ be a polynomial and $\mu^{(1)}_f$, $\mu^{(2)}_f$ be the one- and two-body densities associated to $\Psi_f$ of the form \eqref{eq:restricted_class}. Then, for any test functions $\chi_1$ and $\chi_2$,
\begin{equation} \label{eq:qm_mf_approximation}
\Big|\int_{\mathbb{R}^2}\big(\mu^{(1)}_{f}(x)-\mu^\mathrm{el}_{f}(x)\big)\chi_1(x)dx\Big|\le C\Big(\frac{\log N}{N}\Big)^{1/2}\|\nabla\chi_1\|_{L^2}+CN^{-1/2}\|\nabla\chi_1\|_{L^\infty},
\end{equation}
and
\begin{equation} \label{eq:qm_mf_approximation2}
\begin{split}
\Big|\iint_{\mathbb{R}^2\times\mathbb{R}^2}\big(\mu^{(2)}_{f}(x,y)-\mu^\mathrm{el}_{f}(x)\mu^\mathrm{el}_{f}(y)\big)\chi_2(x,y)dxdy\Big|\le\;& C\Big(\frac{\log N}{N}\Big)^{1/2}\sup_y\|\nabla\chi_2(\cdot,y)\|_{L^2}\\
&+CN^{-1/2}\sup_y\|\nabla\chi_2(\cdot,y)\|_{L^\infty}.
\end{split}
\end{equation}
Moreover, we have the pointwise bounds
\begin{equation} \label{eq:pointwise_bounds}
\begin{split}
0\le\mu^{(1)}_{f}(x)&\le C e^{-NC(|x|^2-\log N)}\\
0\le\mu^{(2)}_{f}(x,y)&\le C e^{-NC(|x|^2+|y|^2-\log N)}
\end{split}
\end{equation}
\end{proposition}

\begin{proof}[Proof of Proposition \ref{prop:mf_upper}]
The bound \eqref{eq:qm_mf_approximation} corresponds exactly to \cite[Eq. (4.1)]{RouYng-17}, which in turn is based on \cite[Theorem 3.2]{RouSerYng-13b} and we refer to those papers for the complete proof. The argument is based on the free-energy functional
\begin{equation}
\mathcal{F}^\mathrm{el}_{f}[\mu]=\mathcal{E}^\mathrm{el}_{f}[\mu]+\frac{1}{N}\int\mu(x)\log\mu(x)dx
\end{equation}
at temperature $N^{-1}$ associated with the electrostatic energy functional $\mathcal{E}^\mathrm{el}_{f}$ \eqref{eq:electrostatic_energy}. Let us denote by $\mu^{\mathrm{Gibbs}}_{f}$ the unique minimizer of $\mathcal{F}^\mathrm{el}_{f}$ among probability measures.

In order to prove \eqref{eq:qm_mf_approximation2}, we notice that an intermediate step for the proof of \eqref{eq:qm_mf_approximation} is the bound
\begin{equation} \label{eq:Coulomb_proximity}
D\left(\mu^{\mathrm{Gibbs}}_{f}-\mu^{\mathrm{el}}_{f},\mu^{\mathrm{Gibbs}}_{f}-\mu^{\mathrm{el}}_{f}\right)\le CN^{-1}.
\end{equation}
This corresponds to \cite[Eq. (4.4)]{RouYng-17}. We also import, directly from \cite[Remark 3.3]{RouSerYng-13b} (see also \cite[Lemma 7.4]{RouSer-14})
\begin{equation} \label{eq:second_marginal1}
\begin{split}
\Big|\iint_{\mathbb{R}^2\times\mathbb{R}^2}\big(\mu^{(2)}_{f}(x,y)-\mu^\mathrm{Gibbs}_{f}(x)\mu^\mathrm{Gibbs}_{f}(y)\big)\chi_2(x,y)dxdy\Big|\le\;& C\Big(\frac{\log N}{N}\Big)^{1/2}\sup_y\|\nabla\chi_2(\cdot,y)\|_{L^2}\\
&+CN^{-1/2}\sup_y\|\nabla\chi_2(\cdot,y)\|_{L^\infty}.
\end{split}
\end{equation}
Moreover, by \eqref{eq:Coulomb_distance2} and \eqref{eq:Coulomb_proximity},
\begin{equation} \label{eq:second_marginal2}
\Big|\iint_{\mathbb{R}^2\times\mathbb{R}^2}\big(\mu^\mathrm{Gibbs}_{f}(x)\mu^\mathrm{Gibbs}_{f}(y)-\mu^\mathrm{el}_{f}(x)\mu^\mathrm{el}_{f}(y)\big)\chi_2(x,y)dxdy\Big|\le C N^{-1/2}\sup_y\|\nabla\chi_2(\cdot,y)\|_{L^2}.
\end{equation}

Hence, \eqref{eq:qm_mf_approximation2} is proven by decomposing
\begin{equation*}
\begin{split}
\iint_{\mathbb{R}^2\times\mathbb{R}^2}&\big(\mu^{(2)}_{f}(x,y)-\mu^\mathrm{el}_{f}(x)\mu^\mathrm{el}_{f}(y)\big)\chi_2(x,y)dxdy\\
=&\iint_{\mathbb{R}^2\times\mathbb{R}^2}\big(\mu^{(2)}_{f}(x,y)-\mu^\mathrm{Gibbs}_{f}(x)\mu^\mathrm{Gibbs}_{f}(y)\big)\chi_2(x,y)dxdy\\
&+\iint_{\mathbb{R}^2\times\mathbb{R}^2}\big(\mu^\mathrm{Gibbs}_{f}(x)\mu^\mathrm{Gibbs}_{f}(y)-\mu^\mathrm{el}_{f}(x)\mu^\mathrm{el}_{f}(y)\big)\chi_2(x,y)dxdy
\end{split}
\end{equation*}
and then using \eqref{eq:second_marginal1} and \eqref{eq:second_marginal2}.

We are only left with the proof of \eqref{eq:pointwise_bounds}, which is obtained analogously to the proof of \cite[Eq. (4.2)]{RouYng-17}. There one first obtains (cf. \cite[Eq. (4.14)]{RouYng-17})
\begin{equation} \label{eq:partial_pointwise_estimate}
0\le \mu^\mathrm{Gibbs}_f(x)\le C e^{-NC|x|^2}.
\end{equation}
This estimate is then carried over to one- and two-particle densities as in the proof of \cite[Eq. (3.16)]{RouSerYng-13b}. More precisely, at the end of Section 3.3 of \cite{RouSerYng-13b} the following formula is obtained
\begin{equation*}
\mu_f(x_1,\dots,x_N)\le e^{CN\log N}\prod_{j=1}^N\mu_f^\mathrm{Gibbs}(x_j).
\end{equation*}
By integrating with respect to all but one or two variables one gets
\begin{equation*}
\begin{split}
\mu_f^{(1)}(x)\le\;& e^{CN\log N}\mu_f^\mathrm{Gibbs}(x)\\
\mu_f^{(2)}(x,y)\le\;& e^{CN\log N}\mu_f^\mathrm{Gibbs}(x)\mu_f^\mathrm{Gibbs}(y),
\end{split}
\end{equation*}
which directly imply \eqref{eq:pointwise_bounds} due to \eqref{eq:partial_pointwise_estimate}.
\end{proof}

Having at hand Proposition \ref{prop:mf_upper} which shows proximity of $\mu^{(1)}_f$ and $\mu^\mathrm{el}_f$, the missing ingredient is a link between $\mu^\mathrm{el}_f$ and the measure $\mu^\mathrm{sol}$ of Theorem \ref{thm:energies up}. This is provided by the following result.

\begin{proposition}[\textbf{Inverse electrostatic problem}]\mbox{}\label{prop:inverse_electrostatic}\\
Let $\mu^\mathrm{sol}$ be, as in Theorem \ref{thm:energies up}, a probability measure whose only values are 0 and $B/2\pi\ell$. There exists a (sequence of) polynomial(s) $f_\delta$ indexed by a $N$-dependent parameter $\delta>0$ such that, for any test-functions $\chi_1$ and $\chi_2$,
\begin{equation} \label{eq:estimate_Coulomb_norm}
\Big|\int_{\mathbb{R}^2}\big(\mu^\mathrm{el}_{f_\delta}(x)-\mu^\mathrm{sol}(x))\chi(x)dx\Big|\le C N^{-1/4}\|\nabla\chi\|_{L^2}
\end{equation}
and
\begin{equation} \label{eq:estimate_Coulomb_norm2}
\begin{split}
\Big|\iint_{\mathbb{R}^2\times\mathbb{R}^2}\big(\mu^\mathrm{el}_{f_\delta}(x)\mu^\mathrm{el}_{f_\delta}(y)-\mu^\mathrm{sol}(x)\mu^\mathrm{sol}(y)\big)\chi_2(x,y)dxdy\Big|\le CN^{-1/4}\sup_y\|\nabla\chi_2(\cdot,y)\|_{L^2}.
\end{split}
\end{equation}
\end{proposition}

\begin{proof}
Proposition \ref{prop:inverse_electrostatic} is proven analogously to \cite[Prop. 3.1]{RouYng-17}, the main difference being that the role of the bathtub minimizer $\rho_0$ is now played by $\mu^{\mathrm{sol}}_{}$. One can anyway repeat all the steps and deduce
\begin{equation*}
D(\mu^\mathrm{sol}-\mu^\mathrm{el}_{f_\delta},\mu^\mathrm{sol}-\mu^\mathrm{el}_{f_\delta}) \le CN^{-1/2}.
\end{equation*}
The result is then brought to the form \eqref{eq:estimate_Coulomb_norm} and \eqref{eq:estimate_Coulomb_norm2} using Lemma \ref{lemma:Coulomb_distance}.
\end{proof}

We can now provide the 

\begin{proof}[Proof of Theorem~\ref{thm:energies up}]
Let us introduce $\chi_\mathrm{in}$ and $\chi_\mathrm{out}$, smooth partition of unity with $\chi_\mathrm{in}$ supported in the disk $D(0,2\log N)$ and $\chi_\mathrm{out}$ identically zero in the disk $D(0,\log N)$. We can choose them so that
\begin{equation*}
\|\nabla \chi_\mathrm{in}\|_\infty+\|\nabla \chi_\mathrm{in}\|_\infty\le C\log N,
\end{equation*}
which implies, by the growth assumption \eqref{eq:growth_trap},
\begin{equation} \label{eq:norms_cutoff}
\begin{split}
\|\nabla(\chi_\mathrm{in}v)\|_{L^p}\le\;& C N^{\theta}\\
\|\nabla(\chi_\mathrm{in}w)\|_{L^p}\le\;& C N^{\theta}
\end{split}
\end{equation}
for any $\theta>0$ and $p\in[2,+\infty]$.

Let $f_\delta$ be the sequence of polynomials provided by Proposition \ref{prop:inverse_electrostatic}. We have
\begin{equation} \label{eq:upper_bound_partial}
\begin{split}
\frac{e (N,\lambda)}{N}\le&\;\frac{\mathcal{E}_{N,\lambda}[\Psi_{f_\delta}]}{N}\\
=&\;\int\chi_\mathrm{in}(x)v(x)\mu^\mathrm{sol}(x)dx+\frac{\lambda}{2}\int \int \chi_\mathrm{in}(x)\chi_\mathrm{in}(x)w(x-y)\mu^\mathrm{sol}(x)\mu^\mathrm{sol}(y)dxdy\\
&+\int \chi_\mathrm{in}(x)v(x)\Big(\mu^{(1)}_{f_\delta}(x)-\mu^\mathrm{sol}_{}(x)\Big)dx\\
&+\frac{\lambda}{2}\int\chi_\mathrm{in}(x)\chi_\mathrm{in}(y) w(x-y)\Big(\mu^{(2)}_{f_\delta}(x,y)-\mu^\mathrm{sol}_{}(x)\mu^\mathrm{sol}_{}(y)\Big)dxdy\\
&+\int \chi_\mathrm{out}(x)v(x)\mu^{(1)}_{f_\delta}(x)dx\\
&+\lambda\int\chi_\mathrm{in}(x)\chi_\mathrm{out}(y) w(x-y)\mu^{(2)}_{f_\delta}(x,y)dxdy\\
&+\frac{\lambda}{2}\int\chi_\mathrm{out}(x)\chi_\mathrm{out}(y) w(x-y)\mu^{(2)}_{f_\delta}(x,y)dxdy.
\end{split}
\end{equation}
Since $\mu^\mathrm{sol}$ is a fixed probability measure with compact support, the first line in the right-hand side of \eqref{eq:upper_bound_partial} coincides with $\mathcal{E}^\mathrm{MF}[\mu^\mathrm{sol}]$ for $N$ large enough. To estimate the second and third lines we use \eqref{eq:qm_mf_approximation}, \eqref{eq:qm_mf_approximation2}, \eqref{eq:estimate_Coulomb_norm}, \eqref{eq:estimate_Coulomb_norm2}, together with \eqref{eq:norms_cutoff}, thus obtaining
\begin{equation*}
\begin{split}
\Big|\int \chi_\mathrm{in}(x)v(x)\Big(\mu^{(1)}_{f_\delta}(x)-\mu^\mathrm{sol}_{}(x)\Big)dx\Big|\le\;& CN^{-1/4+\gamma}\\
\frac{\lambda}{2}\;\Big|\int\chi_\mathrm{in}(x)\chi_\mathrm{in}(y) w(x-y)\Big(\mu^{(2)}_{f_\delta}(x,y)-\mu^\mathrm{sol}_{}(x)\mu^\mathrm{sol}_{}(y)\Big)dxdy\Big|\le\;& CN^{-1/4+\gamma}
\end{split}
\end{equation*}
for any $\gamma>0$. The remaining three terms are estimated using \eqref{eq:pointwise_bounds}. For example
\begin{equation*}
\begin{split}
\Big|\int\chi_\mathrm{out}(x)v(x)\mu^{(1)}_{f_\delta}(x)dx\Big|\le \int_{|x|\ge \log N} e^{-CN(|x|^2-\log N)}|x|^s
\end{split}
\end{equation*}
is exponentially small as $N\to\infty$, and the same for the other terms. We obtained
\begin{equation*}
\frac{e (N,\lambda)}{N}\le \frac{\mathcal{E}_{N,\lambda}[\Psi_{f_\delta}]}{N}\le \mathcal{E}^\mathrm{MF}[\mu^\mathrm{sol}](1+N^{-1/4+\gamma})
\end{equation*}
for any $\gamma>0$, which concludes the proof.
\end{proof}

\section{Convergence of densities} \label{sect:densities}

\subsection{The Hewitt-Savage and Diaconis-Freedman Theorems}

We start by recalling the Hewitt-Savage Theorem~\cite{HewSav-55,DiaFre-80} for $N$-body states, see~\cite{Rougerie-spartacus,Rougerie-LMU,Mischler-11} for more details. We will denote by $\mathcal{P}(\Sigma)$ the set of probability measures over a set $\Sigma$. We shall use the fact that the de Finetti-Hewitt-Savage can be approximated by the Diaconis-Freedman construction: given a symmetric probability measure $\mu_N$ over $\R^{dN}$, $d\geq 1$, we define the probability measure $P_{\mu_N}$ over $\mathcal{P}(\mathbb{R}^2)$ by
\begin{equation}\label{eq:DF_measure}
\int_{\mathcal{P}(\mathbb{R}^2)}\phi(\sigma)dP_{\mu_N}(\sigma):=\int \phi(\mathrm{Emp}_{X_N})d\mu_N(X_N),
\end{equation}
for any $\phi\in C^0(\mathcal{P}(\mathbb{R}^2))$ or, by a slight abuse of notation,
\begin{equation} 
P_{\mu_N}(\sigma) =\int_{\mathbb{R}^{2N}} \delta_{\sigma = \Emp_{X_N}} d\mu_N (x_1,\dots,x_N),
\end{equation}

We then have the

\begin{theorem}[\textbf{Hewitt-Savage in large $N$ limit}]\mbox{} \label{thm:Hewitt-Savage}\\
Let $\mu_N$ be a symmetric probability measure over $\R^{dN}$, and let its $k$-marginal $\mu_N^{(k)}$ be defined by integrating over $N-k$ $d$-dimensional variables. Assume that $\mu^{(1)}_N$ is tight:
\begin{equation} \label{eq:tightness_assumption}
\limsup_{R\to\infty}\sup_{N\in\mathbb{N}}\Big(1-\mu_N^{(1)}\left(B(0,R)\right)\Big)=0.
\end{equation}
Extract a subsequence such that 
$\mu_N ^{(1)} \wto \mu \in \cP (\R^d)$ 
as measures. Then, along this (not-relabeled) subsequence, 
\begin{enumerate}
 \item There exists a unique probability measure $P \in \cP (\cP (\R^d))$ such that 
 for any fixed $k\in\mathbb{N}$,
\begin{equation} \label{eq:HS}
\mu_N^{(k)}\rightharpoonup\int_{\mathcal{P}(\mathbb{R}^d)}\mu^{\otimes k}dP(\mu).
\end{equation}
\item Let $P_{\mu_N}$ be defined in~\eqref{eq:DF_measure} and $P$ the measure such that~\eqref{eq:HS} holds. We have that 
\begin{equation}\label{eq:DF to HS}
P_{\mu_N} \rightharpoonup P 
\end{equation}
weakly as measures on $\cP (\R^d)$.
\end{enumerate}
\end{theorem}

\begin{proof}[Elements of proof]
If $\R^d$ is replaced by a compact subset, the proof is contained in~\cite[Chapter~2]{Rougerie-spartacus,Rougerie-LMU}. The compactness assumption can be removed by a one-point compactification argument that can be found in~\cite{Mischler-11}.

As for Item~2 of the statement, see again~\cite[Chapter~2]{Rougerie-spartacus,Rougerie-LMU} for the case where $\R^d$ is compactified. First extract a subsequence along which 
$$ P_N \wto P'$$
for some probability measure $P'$ on $\cP \left(\overline{\mathbb{R}^d}\right)$. We recall the Diaconis-Freedman bound
\begin{equation}\label{eq:DF}
\Big\|\mu_N^{(k)}-\widetilde\mu_N^{(k)}\Big\|_{\mathrm{TV} }\le 2\frac{k(k-1)}{N}.
\end{equation}
where we used the total variation norm
\begin{equation*}
\|\mu-\nu\|_{\mathrm{TV}}=2\sup\big\{|\mu(A)-\nu(A)|\;|\;A\in\Sigma\big\}.
\end{equation*}
and denoted 
\begin{equation}
\widetilde\mu_N=\int_{\mathcal{P}(\mathbb{R}^d)}\mu^{\otimes N}dP_N (\mu).
\end{equation}
It follows that $P' = P$. Indeed, passing to the limit in~\eqref{eq:DF}, the restriction of $P'$ to $\mathcal{P}({\mathbb{R}^d})$ must be $P$. But the latter is a probability measure, so $P'$ must charge only $\mathcal{P}({\mathbb{R}^d})$.
\end{proof}

\subsection{Energy estimates on the support of the limit measure}

As a preparatory step for the proof of Theorem \ref{thm:densities} we will show in this subsection that the limit measure $P$ which Theorem \ref{thm:Hewitt-Savage} associates to a many-body minimizing sequence only charges probability measures whose energy is smaller than $E^\mathrm{flo}$.

For any $\delta>0$ define the function $U:\mathcal{P}(\mathbb{R}^2)\to\mathbb{R}$ as
\begin{equation}
U_\delta(\mu)=\begin{cases}
1\qquad\text{if }\mathcal{E}^\mathrm{MF}[\mu]>E^\mathrm{flo}+\delta\;\;\text{ or if } \mathcal{E}^\mathrm{MF}[\mu]\text{ is not defined}\\
0\qquad\text{if }\mathcal{E}^\mathrm{MF}[\mu]\le E^\mathrm{flo}+\delta.
\end{cases}
\end{equation}
Notice that $U_\delta$ is also defined on probability measures that do not satisfy the upper constraint $\mu\le B/2\pi\ell$ of the flocking problem \ref{eq:mf_energy}.
The following lemma will be used later in order to apply a weak version of Fatou's lemma.

\begin{lemma} [\textbf{Weak-$\star$ lower semicontinuity of $U_\delta$}]\mbox{}\label{lemma:lower_semicont_U}\\ 
For any sequence of probability measures $\mu_n$ converging weakly-$\star$ as Radon measures to $\mu$ we have
\begin{equation*}
\liminf_{n\to\infty}U_\delta(\mu_n)\ge U_\delta(\mu).
\end{equation*}
\end{lemma}

\begin{proof}
We have mentioned in Section~\ref{sec:exist} that $\mathcal{E}^\mathrm{MF}$ is lower semicontinuous. Hence 
$$U^{-1}_\delta(0)=\{\mu\in\mathcal{P}(\mathbb{R}^2)\,|\,\mathcal{E}^\mathrm{MF}[\mu]\le E^\mathrm{flo}+\delta\}$$
is closed. It follows that $\{\mu\in\mathcal{P}(\mathbb{R}^2)\,|\, U_\delta(\mu)\le\alpha\}\subset\mathcal{P}(\mathbb{R}^2)$ is closed for any $\alpha\in\mathbb{R}$, which is lower semicontinuity of $U_\delta$.
%
%
%
\end{proof}

%

\begin{proposition}[\textbf{The limit measure of minimizing sequences charges low energies}]\mbox{}\label{prop:HS_measure}\\ 
Let $\Psi_F$ be a minimizing sequence for the many-body energy \eqref{eq:many_body_energy}, i.e.,
\begin{equation*}
\mathcal{E}_{N,\lambda}[\Psi_F]=E (N,\lambda)+o(N)
\end{equation*}
as $N\to\infty$. Let $P$ be the limit measure associated to $|\Psi_F|^2$ via Theorem~\ref{thm:Hewitt-Savage}. Then
\begin{equation*}
\mathcal{E}^\mathrm{MF}[\mu]\le E^\mathrm{flo}\quad P\;\text{ - a.e.}
\end{equation*}
\end{proposition}

\begin{proof}
We want to apply Theorem \ref{thm:Hewitt-Savage} to $|\Psi_F|^2$.  We only have to prove the tightness condition~\eqref{eq:tightness_assumption} for $\mu^{(1)}_{\Psi_F}$. But, for $R$ large enough,
\begin{equation*}
1-\mu_{\Psi_F}^{(1)}(B(0,R))=\int_{\mathbb{R}^2\setminus B(0,R)}\frac{v(x)}{v(x)}\mu_{\Psi_F}^{(1)}(x)dx \le \frac{1}{v(R)}\left(N^{-1}\mathcal{E}_{N,\lambda}[\Psi_F]\right)\le \frac{E^\mathrm{flo}+C}{v(R)}.
\end{equation*}
The first inequality follows from the growth of $v$ in~\eqref{eq:growth_trap} and the boundedness of $w$. The second one is clear since we proved in the previous sections that
\begin{equation*}
\lim_{N\to\infty}\frac{E(N,\lambda)}{N}=E^\mathrm{flo}
\end{equation*}
for $\lambda$ small enough.

Denote now $P_{\Psi_F}$ the measure over $\cP (\R^2)$ associated to $|\Psi_F| ^2$ via~\eqref{eq:DF_measure}.  We first claim that, for any fixed $\delta>0$, 
\begin{equation}\label{eq:DF meas sup}
\int_{\mathcal{P}(\mathbb{R}^2)}U_\delta(\mu)dP_{\Psi_F}(\mu)\le (\delta^{-1}+1)o_N(1).
\end{equation}
Let us define the set
\begin{equation*}
\Xi_\delta=\big\{X_N\in\mathbb{R}^{2N}\;|\;\mathcal{E}^\mathrm{MF}[\Emp_{X_N}]> E^\mathrm{flo}+\delta\big\}.
\end{equation*}
From the definition~\eqref{eq:DF_measure} of $P_{\Psi_F}$ we deduce
\begin{equation} \label{eq:DF_U_equality}
\int_{\mathcal{P}(\mathbb{R}^2)}U_\delta(\mu)dP_{\Psi_F}(\mu)=\mu_F(\Xi_\delta).
\end{equation}
In order to obtain an estimate for $\mu_F(\Xi_\delta)$, our proof now goes through computations similar to those in the proof of \eqref{eq:lower_bound} at the end of Section \ref{sect:lower_bound}. In particular, importing directly \eqref{eq:first_splitting}, we have
\begin{equation} \label{eq:splitting_xi}
\begin{split}
N^{-1}\mathcal{E}_{N,\lambda}[\Psi_F]\ge\;&\int_{A_L\cap \Xi_\delta}\mathcal{E}^{\mathrm{MF}}[\Emp_{X_N}]d\mu_F(X_N)+\int_{A_L\cap \Xi_\delta^c}\mathcal{E}^{\mathrm{MF}}[\Emp_{X_N}]d\mu_F(X_N)\\
&-CN^{-1}.
\end{split}
\end{equation}
where
\begin{equation*}
A_L=\Big\{X_N\in \mathbb{R}^{2N}\,|\,\int_{S_k}\Emp_{X_N}\le (1+\varepsilon)\frac{B|S_k|}{2\pi\ell}\quad\forall k\,\text{ such that }\,S_k\subset\mathcal{S}_L\Big\},
\end{equation*}
$\mathcal{S}_L$ is the square $[-L,L]\times[-L,L]$ and $S_k$ is the $k$-th square of side $N^\alpha$ of a tiling of $\mathcal{S}_L$. We now use the definition of $\Xi_\delta$ in the integral on $A_L\cap \Xi_\delta$ to get
\begin{equation*}
\int_{A_L\cap \Xi_\delta}\mathcal{E}^{\mathrm{MF}}[\Emp_{X_N}]d\mu_F(X_N)> (E^\mathrm{flo}+\delta)\mu_F(A_L\cap\Xi_\delta).
\end{equation*}
Notice also that
\begin{equation*}
\mu_F(A_L\cap \Xi_\delta)=\mu_F(\Xi_\delta)-\mu_F(A_L^c\cap \Xi_\delta)\ge \mu_F(\Xi_\delta)-\mu_F(A_L^c),
\end{equation*}
and we can use \eqref{eq:union_bound} to estimate $\mu(A_L^c)$. Furthermore, in the integral on $A_L\cap \Xi_\delta^c$, we perform a Riemann sum approximation as explained in the proof of \eqref{eq:lower_bound} when passing from \eqref{eq:first_splitting} to \eqref{eq:second_step}. After repeating the same steps, this procedure yields
\begin{equation*}
\int_{A_L\cap \Xi_\delta^c}\mathcal{E}^{\mathrm{MF}}[\Emp_{X_N}]d\mu_F(X_N)\ge E^\mathrm{flo}\mu_F(A_L\cap \Xi_\delta^c)-C\varepsilon-CN^\alpha.
\end{equation*}
Hence, we can bring \eqref{eq:splitting_xi} to the form
\begin{equation*}
N^{-1}\mathcal{E}_{N,\lambda}[\Psi_F]\ge E^\mathrm{flo}+\delta \mu_F(\Xi_\delta)-C\delta L^2N^{-2\alpha}e^{-C\varepsilon N^{\sqrt{5}}-1}-CN^\alpha-C\varepsilon.
\end{equation*}
Since we chose a minimizing sequence $\Psi_F$ we have, by the upper bound \eqref{eq:upper_bound},
\begin{equation*}
N^{-1}\mathcal{E}_{N,\lambda}[\Psi_F] \le E^\mathrm{flo}+o_N(1)
\end{equation*}
as $N\to\infty$. Combining the last two inequalities we find that we can choose $\varepsilon=N^{-1}$ and $L$ large enough such that
\begin{equation*}
\delta \mu_F(\Xi_\delta)\le (1+\delta)o_N(1).
\end{equation*}
Using~\eqref{eq:DF_U_equality}, this completes the proof of~\eqref{eq:DF meas sup}.

Next, since $P_{\Psi_F}$ converges weakly as measures to $P$ as in \eqref{eq:DF to HS}, the integral
\begin{equation*}
\int_{\mathcal{P}(\mathbb{R}^2)}U_\delta(\mu)dP_{\Psi_F}(\mu)
\end{equation*}
has a positive and weakly lower semicontinuous function $U_\delta$ (see Lemma \ref{lemma:lower_semicont_U}) integrated against a weakly converging sequence of measures. We are within the assumptions of \cite[Chapter 1, Section 2, Problem 2.6]{Bil-99} and \cite[Theorem 1.1]{FeiKasZad-14} which yield the Fatou-like lemma
\begin{equation*}
\liminf_{N\to\infty} \int_{\cP(\mathbb{R}^2)}U_\delta(\mu)dP_{\Psi_F}(\mu)\ge \int_{\mathcal{P}(\mathbb{R}^2)}U_\delta(\mu)dP(\mu),
\end{equation*}
whence
\begin{equation*}
\int_{\mathcal{P}(\mathbb{R}^2)}U_\delta(\mu)dP_{\Psi_F}(\mu)\ge \int_{\mathcal{P}(\mathbb{R}^2)}U_\delta(\mu)dP(\mu)-o_N(1).
\end{equation*}
Combining with~\eqref{eq:DF meas sup} we deduce
\begin{equation*}
\int_{\mathcal{P}(\mathbb{R}^2)}U_\delta(\mu)dP(\mu)\le (1+\delta^{-1})o_N(1),
\end{equation*}
whence, passing to the limit $N\to\infty$,
\begin{equation*}
\int_{\mathcal{P}(\mathbb{R}^2)}U_\delta(\mu)dP(\mu)\le 0
\end{equation*}
The subsequent limit $\delta\to 0$ proves the desired statement.
\end{proof}

\subsection{Incompressibility estimate for the limit measure and proof of Theorem \ref{thm:densities}}

\begin{proposition}[\textbf{Incompressibility estimate for the limit measure}]\mbox{}\label{prop:incompressibility_HS}\\ 
Let $\Psi_F$ be a minimizing sequence for the many-body energy \eqref{eq:many_body_energy}, i.e.,
\begin{equation*}
\mathcal{E}_{N,\lambda}[\Psi_F]=E (N,\lambda)+o(N)
\end{equation*}
as $N\to\infty$. Let $P\in \cP (\cP (\R^2))$ be the limit measure associated to $|\Psi_F| ^2$ by Theorem \ref{thm:Hewitt-Savage}. Then, $P$-almost everywhere, 
\begin{equation*}
\mu \le \frac{B}{2\pi\ell}.
\end{equation*}
\end{proposition}

\begin{proof}
For any $k\in\mathbb{N}$ and $\Omega\subset\mathbb{R}^2$, define the set
\begin{equation*}
J_{k,\Omega}=\left\{X_N\in\mathbb{R}^{2N}\,\Big|\, \left(\int_{\Omega}\Emp_{X_N}\right)^k>\left((1+\varepsilon)\frac{B|\Omega|}{2\pi\ell}\right)^k \right\}.
\end{equation*}
Without loss of generality, since balls form a basis for Borel subsets of $\mathbb{R}^2$, we will assume that $\Omega$ is a ball. 
Since $J_{k,\Omega}=J_{1,\Omega}$, Theorem~\ref{thm:exponential_bound} implies immediately
\begin{equation*}
\mu_F(J_{k,\Omega})\le  e^{-C\varepsilon N^{\sqrt{5}-1}}.
\end{equation*}
Notice that the $C$ on the right of the last inequality depends on the set $\Omega$, which we will however keep fixed. We write
\begin{equation} \label{eq:estimate_DF_omega}
\begin{split}
\int_{\mathbb{R}^{2N}}\Big(\int_\Omega\Emp_{X_N}\Big)^kd\mu_F(X_N)=\;&\int_{J_{k,\Omega}^c}\Big(\int_\Omega\Emp_{X_N}\Big)^kd\mu_F(X_N)\\
&+\int_{J_{k,\Omega}}\Big(\int_\Omega\Emp_{X_N}\Big)^kd\mu_F(X_N)\\
\le\;&\Big((1+\varepsilon)\frac{B|\Omega|}{2\pi\ell}\Big)^k+e^{-C\varepsilon N^{\sqrt{5}-1}},
\end{split}
\end{equation}
having used $\int_\Omega\Emp_{X_N}\le1$ for the second term.

Let $P$ be the measure on $\mathcal{P}(\mathbb{R}^2)$ that Theorem \ref{thm:Hewitt-Savage} associates to $\mu_F$. We will show that
\begin{equation} \label{eq:incompressibility_HS}
\int_{\mathcal{P}(\mathbb{R}^2)}\Big(\int_{\Omega}\mu\Big)^kdP(\mu)\le \Big(\frac{B|\Omega|}{2\pi\ell}\Big)^k.
\end{equation}
Let us denote by $\chi_\Omega$ the characteristic function of $\Omega$, and let $\chi_{\Omega,\vartheta}$ be its regularization on a small scale $\vartheta$. Let $\Omega_\vartheta$ be the dilation of $\Omega$ by a factor $(1+\vartheta)$ and let $\chi_{\Omega_\vartheta}$ be the corresponding characteristic function. We can pick these functions so as to have
\begin{equation} \label{eq:construction_balls}
\chi_\Omega\le \chi_{\Omega,\vartheta}\le \chi_{\Omega_\vartheta}.
\end{equation}
Then
\begin{equation*}
\begin{split}
\int_{\mathcal{P}(\mathbb{R}^2)}\Big(\int_\Omega\mu\Big)^kdP(\mu)\le\;& \int_{\mathcal{P}(\mathbb{R}^2)}\Big(\int\chi_{\Omega,\vartheta}\,\mu\Big)^kdP(\mu)\\
=\;&\int \chi^{\otimes k}_{\Omega,\vartheta}\: \mu_F^{(k)}+o_N(1)\\
\le\;&\int_{\Omega_\vartheta^{k}}\mu_F^{(k)}+o_N(1),
\end{split}
\end{equation*}
where the first step is due to \eqref{eq:construction_balls}, the second is due to Theorem~\ref{thm:Hewitt-Savage} (and the bound~\eqref{eq:DF}, notice that $\chi_{\Omega,\delta}$ is continuous and compactly supported), and the third again to \eqref{eq:construction_balls}. We remark that the error $o_N(1)$ might depend badly on $\vartheta$. Now, by~\eqref{eq:DF} we have
\begin{equation*}
\int_{\Omega_\vartheta^{k}}\mu_F^{(k)}=\int_{\Omega_\vartheta^{k}}\widetilde\mu_F^{(k)}+o_N(1).
\end{equation*}
Furthermore, notice that
\begin{equation*}
\int_{\Omega_\vartheta^{\times k}}\widetilde\mu_F^{(k)}=\int_{\mathbb{R}^{2N}}\Big(\int_{\Omega_\vartheta}\Emp_{X_N}\Big)^kd\mu_F(X_N).
\end{equation*}
By putting together the last three estimates we get
\begin{equation*}
\begin{split}
\int_{\mathcal{P}(\mathbb{R}^2)}\Big(\int_{\Omega}\mu\Big)^kdP(\mu)\le\;& \int_{\mathbb{R}^{2N}}\Big(\int_{\Omega_\vartheta}\Emp_{X_N}\Big)^kd\mu_F(X_N)+o_N(1)\\
\le\;&\Big((1+\varepsilon)\frac{B|{\Omega_\vartheta}|}{2\pi\ell}\Big)^k+e^{-C\varepsilon N^{\sqrt{5}-1}}+o_N(1),
\end{split}
\end{equation*}
having used \eqref{eq:estimate_DF_omega} for the second step. We can choose, for example, $\varepsilon=N^{-1}$, which allows to take the limit $N\to\infty$ first. The subsequent limit $\vartheta\to0$ finally yields \eqref{eq:incompressibility_HS}.

The proof is then completed by an argument already appearing in \cite[Theorem 2.6]{FouLewSol-15}. Assume for contradiction that there exists a set $\mathcal{K}\subset\mathcal{P}(\mathbb{R}^2)$ such that ${P}(\mathcal{K})>0$ and $\mu(\Omega)>|\Omega|B/2\pi\ell$ for some ball $\Omega\subset\mathbb{R}^2$ and for every $\mu\in\mathcal{K}$. We have, by \eqref{eq:incompressibility_HS},
\begin{equation*}
\int_{\mathcal{K}}\left(\frac{\int_{\Omega}\mu}{\frac{B|\Omega|}{2\pi\ell}}\right) ^k dP(\mu)\le 1.
\end{equation*}
Since
$$ \int_{\Omega}\mu > \frac{B|\Omega|}{2\pi\ell}$$
on the whole of $\mathcal{K}$, the above inequality is absurd for large enough $k$ and therefore it must be that
\begin{equation*}
\mu(\Omega)\le\frac{B|\Omega|}{2\pi\ell}\quad P\;\text{ - a.e.}.
\end{equation*}
This is actually equivalent to having the density of $\mu$ bounded from above by $B/2\pi\ell$ almost everywhere, and the proof is complete.
\end{proof}

We are finally ready to prove Theorem \ref{thm:densities}.

\begin{proof}[Proof of Theorem \ref{thm:densities}]
Let $\Psi_F$ be a minimizing sequence as in \eqref{eq:minimizing_sequence}. By Theorem \ref{thm:Hewitt-Savage} there exists $P$ such that, for every $k\in\mathbb{N}$,
\begin{equation*}
\mu_F^{(k)}\rightharpoonup\int_{\mathcal{P}(\mathbb{R}^2)}\mu^{\otimes k}dP(\mu)
\end{equation*}
weakly as measures as $N\to\infty$. The fact that $P$ is a probability measure follows from the second part of Theorem \ref{thm:Hewitt-Savage} if we show that the sequences of densities $\mu_F^{(k)}$ are tight. For fixed $\varepsilon>0$, let $B(0,R_\varepsilon)$ be the ball of radius $R_\varepsilon$ and centered at the origin. We have
\begin{equation*}
\mu^{(2)}_F(B(0,R_\varepsilon)\times B(0,R_\varepsilon))\ge 1-2\int_{B(0,R_\varepsilon)^c}\mu^{(1)}_F.
\end{equation*}
However, since by our assumption \eqref{eq:growth_trap} the external potential $v$ grows at infinity, we have
\begin{equation*}
\int_{B(0,R_\varepsilon)^c}\mu^{(1)}_F\le\frac{1}{v(R_\varepsilon)}\int_{B(0,R_\varepsilon)^c}v\mu^{(1)}_F\le \frac{E^\mathrm{flo}+C}{v(R_\varepsilon)}
\end{equation*}
having used in the second step that (recall that $w\in L^\infty$)
\begin{equation*}
\iint_{\mathbb{R}^2\times\mathbb{R}^2}w(x-y)\mu_F^{(2)}(x,y)dxdy \ge -C
\end{equation*}
together with the fact that the sequence $\mu_F$ minimizes the energy. Hence, for $R_\varepsilon$ large enough, $(E^\mathrm{flo}+C)/v(R_\varepsilon)<\varepsilon$ and we have
\begin{equation*}
\mu^{(2)}_F(B(0,R_\varepsilon)\times B(0,R_\varepsilon))\ge 1-2\varepsilon
\end{equation*}
for all $N$. Tightness of all other marginals follows.

By Proposition \ref{prop:HS_measure} and Proposition \ref{prop:incompressibility_HS} we know that $P$ must be concentrated on
\begin{equation*}
\left\{\mu\;|\;\mathcal{E}^\mathrm{MF}[\mu]\le E^\mathrm{flo},\int_{\mathbb{R}^2}\mu=1\right\}\cap\left\{\mu\;|\;0\le\mu\le \frac{B}{2\pi\ell}\;\text{ a.e. },\;\int_{\mathbb{R}^2}\mu=1\right\}.
\end{equation*}
This is equivalent to saying that $P$ is concentrated on the set of minimizers of the flocking problem \eqref{eq:mf_energy}, which completes the proof.
\end{proof}

%

%

\end{document}